\documentclass[12pt,a4paper]{article}
\usepackage[ansinew]{inputenc}
\usepackage[T1]{fontenc}
\usepackage{times}
\usepackage[english]{babel}
\usepackage{tabularx}
\usepackage{fancybox}
\usepackage{fancyhdr}
\usepackage{graphicx}
\usepackage{color}
\usepackage{amsmath}
\usepackage{amstext}
\usepackage{amssymb}
\usepackage{amsthm}
\usepackage{mathrsfs}
\newtheorem{Thm}{\bf Theorem}[section]
\newtheorem{Ass}{\bf Assumption}

\newtheorem{Def}[Thm]{\bf Definition}
\newtheorem{Cor}[Thm]{\bf Corollary}
\newtheorem{Lem}[Thm]{\bf Lemma}
\newtheorem{Prop}[Thm]{\bf Proposition}
\newtheorem*{Prob}{\bf Problem}
\theoremstyle{remark}
\newtheorem{Rem}[Thm]{\bf Remark}
\newtheorem*{Ex}{\bf Example}
\newcommand{\be}{\begin{equation}}
\newcommand{\ee}{\end{equation}}
\newcommand{\ba}{\begin{aligned}}
\newcommand{\ea}{\end{aligned}}
\newcommand{\R}{\mathbb{R}}
\newcommand{\N}{\mathbb{N}}
\newcommand{\E}{E}

\newcommand{\ind}{1{\hskip -3 pt}\hbox{I}}

\newcommand{\F}{\mathcal{F}}
\newcommand{\FF}{\mathbb{F}}

\DeclareMathOperator{\diag}{diag}

\DeclareMathOperator{\sign}{sign}

\allowdisplaybreaks[4]
\usepackage[labelfont=bf,format=hang]{caption}
\usepackage{indentfirst}
\usepackage{geometry}
\makeatletter
\renewcommand*\@fnsymbol[1]{\the#1}
\makeatother
\newcommand{\Keywords}[1]{\par\noindent{\small{\textbf{Key words\/:} #1}}}

\title{Diffusion-based models for financial markets \\ without martingale measures} 
\date{}

\author{Claudio Fontana\footnote{INRIA Paris - Rocquencourt, Domaine de Voluceau, Rocquencourt, BP 105, 78153 Le Chesnay Cedex (France), e-mail: \texttt{claudio.fontana@inria.fr}}
\and
Wolfgang J. Runggaldier\footnote{University of Padova, Department of Pure and Applied Mathematics, via Trieste 63, I-35121 Padova (Italy), e-mail: \texttt{runggal@math.unipd.it}}}

\begin{document}

\maketitle

\begin{abstract}
\noindent
In this paper we consider a general class of diffusion-based models and show that, even in the absence of an \emph{Equivalent Local Martingale Measure}, the financial market may still be viable, in the sense that strong forms of arbitrage are excluded and portfolio optimisation problems can be meaningfully solved. Relying partly on the recent literature, we provide necessary and sufficient conditions for market viability in terms of the \emph{market price of risk} process and \emph{martingale deflators}. Regardless of the existence of a martingale measure, we show that the financial market may still be complete and contingent claims can be valued under the original (\emph{real-world}) probability measure, provided we use as numeraire the \emph{Growth-Optimal Portfolio}.
\end{abstract}

\vspace{1cm}

\Keywords{arbitrage, hedging, contingent claim valuation, market price of risk, martingale deflator, growth-optimal portfolio, numeraire portfolio, market completeness, utility indifference valuation, benchmark approach.}

\clearpage

\setcounter{section}{-1}
\section{Introduction}

The concepts of \emph{Equivalent (Local) Martingale Measure} (E(L)MM), \emph{no-arbitrage} and \emph{risk-neutral pricing} can be rightfully considered as the cornerstones of modern mathematical finance. It seems to be almost folklore that such concepts can be regarded as mutually equivalent. In fact, most practical applications in quantitative finance are directly formulated under suitable assumptions which ensure that those concepts are indeed equivalent. 

In recent years, maybe due to the dramatic turbulences raging over financial markets, an increasing attention has been paid to models that allow for financial market anomalies. More specifically, several authors have studied market models where stock price bubbles may occur (see e.g. \cite{CH}, \cite{HLW}, \cite{H}, \cite{JPS1}, \cite{JPS2}). It has been shown that bubble phenomena are consistent with the classical no-arbitrage theory based on the notion of \emph{No Free Lunch with Vanishing Risk} (NFLVR), as developed in \cite{DS} and \cite{DS2}. However, in the presence of a bubble, discounted prices of risky assets are, under a risk-neutral measure, \emph{strict} local martingales, i.e. local martingales which are not true martingales. This fact already implies that several well-known and classical results (for instance the \emph{put-call parity} relation, see e.g. \cite{CH}) of mathematical finance do not hold anymore and must be modified accordingly.

A decisive step towards enlarging the scope of financial models has been represented by the study of models which do not fit at all into the classical no-arbitrage theory based on (NFLVR). Indeed, several authors (see e.g. \cite{CL}, \cite{DS3}, \cite{H}, \cite{KK}, \cite{LW}) have studied instances where an ELMM may fail to exist. More specifically, financial models that do not admit an ELMM appear in the context of \emph{Stochastic Portfolio Theory} (see \cite{FK} for a recent overview) and in the \emph{Benchmark Approach} (see the monograph \cite{PH} for a detailed account).
In the absence of a well-defined ELMM, many of the classical results of mathematical finance seem to break down and one is led to ask whether there is still a meaningful way to proceed in order to solve the fundamental problems of portfolio optimisation and contingent claim valuation. It is then a remarkable result that a satisfactory theory can be developed even in the absence of an ELMM, especially in the case of a complete financial market model, as we are going to illustrate.

The present paper aims at carefully analysing a general class of diffusion-based financial models, without relying on the existence of an ELMM. More specifically, we discuss several notions of no-arbitrage that are weaker than the traditional (NFLVR) condition and we study necessary and sufficient conditions for their validity. We show that the financial market may still be viable, in the sense that strong forms of arbitrage are banned from the market, even in the absence of an ELMM. In particular, it turns out that the viability of the financial market is fundamentally linked to a square-integrability property of the \emph{market price of risk} process. Some of the results that we are going to present have already been obtained, also in more general settings (see e.g. \cite{CL}, Chapter 4 of \cite{Fo}, \cite{HS}, \cite{KK}, \cite{Ka1} and \cite{Ka2}). However, by exploiting the It\^o-process structure, we are able to provide simple and transparent proofs, highlighting the key ideas behind the general theory. We also discuss the connections to the \emph{Growth-Optimal Portfolio} (GOP), which is shown to be the unique portfolio possessing the \emph{numeraire} property. In similar diffusion-based settings, related works that study the question of market viability in the absence of an ELMM include \cite{FK}, \cite{RG}, \cite{HLW}, \cite{LW}, \cite{Lo}, \cite{P0} and \cite{Ru}.

\clearpage
Besides studying the question of market viability, a major focus of this paper is on the valuation and hedging of contingent claims in the absence of an ELMM. In particular, we argue that the concept of \emph{market completeness}, namely the capability to replicate every contingent claim, must be kept distinct from the existence of an ELMM. Indeed, we prove that the financial market may be viable and complete regardless of the existence of an ELMM. We then show that, in the context of a complete financial market, there is a unique natural candidate for the price of an arbitrary contingent claim, given by its GOP-discounted expected value under the original (\emph{real-world}) probability measure. To this effect, we revisit some ideas originally appeared in the context of the \emph{Benchmark Approach}, providing more careful proofs and extending some previous results.

The paper is structured as follows. 
Section \ref{S1.1} introduces the general setting, which consists of a class of It\^o-process models satisfying minimal technical conditions. We introduce a basic standing assumption and we carefully describe the set of admissible trading strategies. The question of whether (properly defined) arbitrage opportunities do exist or not is dealt with in Section \ref{S1.1bis}. In particular, we explore the notions of \emph{increasing profit} and \emph{arbitrage of the first kind}, giving necessary and sufficient conditions for their absence from the financial market. In turn, this lead us to introduce the concept of \emph{martingale deflators}, which can be regarded as weaker counterparts to the traditional (density processes of) martingale measures. Section \ref{S2} proves the existence of an unique \emph{Growth-Optimal} strategy, which admits an explicit characterization and also generates the \emph{numeraire portfolio}. In turn, the latter is shown to be the reciprocal of a martingale deflator, thus linking the numeraire portfolio to the no-arbitrage criteria discussed in Section \ref{S1.1bis}. Section \ref{S1.2} starts with the hedging and valuation of contingent claims, showing that the financial market may be complete even in the absence of an ELMM. Section \ref{S4} deals with contingent claim valuation according to three alternative approaches: \emph{real-world pricing}, \emph{upper-hedging pricing} and \emph{utility indifference valuation}. In the particular case of a complete market, we show that they yield the same valuation formula. Section 6 concludes by pointing out possible extensions and further developments.

\section{The general setting}	\label{S1.1}

Let $\left(\Omega,\F,P\right)$ be a complete probability space. For a fixed time horizon $T\in\left(0,\infty\right)$, let $\FF=\left(\F_t\right)_{0\leq t \leq T}$ be a filtration on $\left(\Omega,\F,P\right)$ satisfying the usual conditions of right-continuity and completeness. Let $W=\left(W_t\right)_{0\leq t \leq T}$ be an $\R^d$-valued Brownian motion on the filtered probability space $\left(\Omega,\F,\FF,P\right)$. To allow for greater generality we do not assume from the beginning that $\FF=\FF^W$, meaning that the filtration $\FF$ may be strictly larger than the $P$-augmented Brownian filtration $\FF^W$. Also, the initial $\sigma$-field $\F_0$ may be strictly larger than the trivial $\sigma$-field.

We consider a \emph{financial market} composed of $N+1$ securities $S^i$, for $i=0,1,\ldots,N$, with $N\leq d$. As usual, we let $S^0$ represent a locally riskless asset, which we name \emph{savings account}, and we define the process $S^0=\left(S_t^0\right)_{0\leq t \leq T}$ as follows:
\be	\label{E1-1}
S_t^0 := \exp\left(\int_0^t r_u\,du\right)
\qquad \text{ for } t\in\left[0,T\right]
\ee
where the \emph{interest rate} process $r=\left(r_t\right)_{0\leq t \leq T}$ is a real-valued progressively measurable process such that $\int_0^T\left|r_t\right|dt<\infty$ $P$-a.s.
The remaining assets $S^i$, for $i=1,\ldots,N$, are supposed to be \emph{risky assets}. For $i=1,\ldots,N$, the process $S^i=\left(S^i_t\right)_{0\leq t \leq T}$ is given by the solution to the following SDE:
\be	\label{E1-2}
dS_t^i = S_t^i\,\mu_t^i\,dt + \sum_{j=1}^d S_t^i\,\sigma_t^{i,j}\,dW_t^j	
\qquad\quad
S_0^i = s^i 
\ee 
where:
\begin{itemize}
\item[\emph{(i)}] 
$s^i\in\left(0,\infty\right)$ for all $i=1,\ldots,N$;
\item[\emph{(ii)}] 
$\mu=\left(\mu_t\right)_{0\leq t \leq T}$ is an $\R^N$-valued progressively measurable process with $\mu_t=\left(\mu_t^1,\ldots,\mu_t^N\right)'$ and satisfying $\sum_{i=1}^N\int_0^T\left|\mu_t^i\right|dt<\infty$ $P$-a.s.;
\item[\emph{(iii)}] 
$\sigma=\left(\sigma_t\right)_{0\leq t \leq T}$ is an $\R^{N\times d}$-valued progressively measurable process with $\sigma_t=\left\{\sigma_t^{i,j}\right\}_{\substack{i=1,\ldots,N\\j=1,\ldots,d}}$ and satisfying $\sum_{i=1}^N\sum_{j=1}^d\int_0^T\left(\sigma_t^{i,j}\right)^2dt<\infty$ $P$-a.s.
\end{itemize}
The SDE \eqref{E1-2} admits the following explicit solution, for every $i=1,\ldots,N$ and $t\in\left[0,T\right]$:
\be	\label{E1-3}
S_t^i = s^i\exp\left(\int_0^t\biggl(\mu_u^i-\frac{1}{2}\sum_{j=1}^d\left(\sigma_u^{i,j}\right)^2\biggr)du
+\sum_{j=1}^d\int_0^t\sigma_u^{i,j}\,dW_u^j\right)
\ee
Note that conditions \emph{(ii)}-\emph{(iii)} above represent minimal conditions in order to have a meaningful definition of the ordinary and stochastic integrals appearing in \eqref{E1-3}. Apart from these technical requirements, we leave the stochastic processes $\mu$ and $\sigma$ fully general. For $i=0,1,\ldots,N$, we denote by $\bar{S}^i=\left(\bar{S}_t^i\right)_{0\leq t \leq T}$ the \emph{discounted price} process of the $i$-th asset, defined as $\bar{S}_t^i:=S^i_t/S^0_t$ for $t\in\left[0,T\right]$.

Let us now introduce the following standing Assumption, which we shall always assume to be satisfied without any further mention. 

\begin{Ass}	\label{A1-1}
For all $t\in\left[0,T\right]$, the $\left(N\times d\right)$-matrix $\sigma_t$ has $P$-a.s. full rank.
\end{Ass}

\begin{Rem}
From a financial perspective, Assumption \ref{A1-1} means that the financial market does not contain redundant assets, i.e. there does not exist a non-trivial linear combination of $\left(S^1,\ldots,S^N\right)$ that is locally riskless, in the sense that its dynamics are not affected by the Brownian motion $W$. However, we want to point out that Assumption \ref{A1-1} is only used in the following for proving uniqueness properties of trading strategies and, hence, could also be relaxed.
\end{Rem}

In order to rigorously describe the activity of trading in the financial market, we now introduce the concepts of \emph{trading strategy} and \emph{discounted portfolio process}. In the following Definition we only consider \emph{self-financing} trading strategies which generate positive portfolio processes.

\begin{Def}	\label{D1-1} $ $
\begin{itemize}
\item[(a)]
An $\R^N$-valued progressively measurable process $\pi=\left(\pi_t\right)_{0\leq t \leq T}$ is an \emph{admissible trading strategy} if $\int_0^T\left\|\sigma_t'\,\pi_t\right\|^2dt<\infty$ $P$-a.s. and $\int_0^T\left|\pi_t'\left(\mu_t-r_t\mathbf{1}\right)\right|dt<\infty$ $P$-a.s., where $\mathbf{1}:=\left(1,\ldots,1\right)'\in\R^N$.
We denote by $\mathcal{A}$ the set of all admissible trading strategies.
\item[(b)]
For any $\left(v,\pi\right)\in\R_+\times\mathcal{A}$, the associated \emph{discounted portfolio process} $\bar{V}^{\,v,\pi}=\left(\bar{V}_t^{\,v,\pi}\right)_{0\leq t \leq T}$ is defined by:
\be	\label{E1-5}	\ba
\bar{V}_t^{\,v,\pi} &:=
v\,\mathcal{E}\left(\sum_{i=1}^N\int\pi^i\,\frac{d\bar{S}^i}{\bar{S}^i}\right)_t	\\
&= v\,\exp\left(\int_0^t\pi_u'\left(\mu_u-r_u\mathbf{1}\right)du-\frac{1}{2}\int_0^t\left\|\sigma_u'\,\pi_u\right\|^2du
+\int_0^t\pi_u'\,\sigma_u\,dW_u\right)
\ea	\ee
for all $t\in\left[0,T\right]$, where $\mathcal{E}\left(\cdot\right)$ denotes the stochastic exponential (see e.g. \cite{RY}, Section IV.3).
\end{itemize}
\end{Def}

The integrability conditions in part \emph{(a)} of Definition \ref{D1-1} ensure that both the ordinary and the stochastic integrals appearing in \eqref{E1-5} are well-defined.
For all $i=1,\ldots,N$ and $t\in\left[0,T\right]$, $\pi_t^i$ represents the proportion of wealth invested in the $i$-th risky asset $S^i$ at time $t$. Consequently, $1-\pi_t'\mathbf{1}$ represents the proportion of wealth invested in the savings account $S^0$ at time $t$. Note that part \emph{(b)} of Definition \ref{D1-1} corresponds to requiring the trading strategy $\pi$ to be \emph{self-financing}. 
Observe that Definition \ref{D1-1} implies that, for any $\left(v,\pi\right)\in\R_+\times\mathcal{A}$, we have $V_t^{\,v,\pi}=v\,V_t^{\,1,\pi}$, for all $t\in\left[0,T\right]$. Due to this scaling property, we shall often let $v=1$ without loss of generality, denoting $V^{\,\pi}:=V^{\,1,\pi}$ for any $\pi\in\mathcal{A}$. 
By definition, the discounted portfolio process $\bar{V}^{\,\pi}$ satisfies the following dynamics:
\be	\label{E1-8}
d\bar{V}_t^{\,\pi} = 
\bar{V}_t^{\,\pi}\sum_{i=1}^N\,\pi_t^i\,\frac{d\bar{S}_t^i}{\bar{S}_t^i}
= \bar{V}_t^{\,\pi}\,\pi_t'\left(\mu_t-r_t\mathbf{1}\right)dt + \bar{V}_t^{\,\pi}\,\pi_t'\,\sigma_t\,dW_t
\ee

\begin{Rem}
The fact that admissible portfolio processes are uniformly bounded from below by zero excludes pathological \emph{doubling strategies} (see e.g. \cite{KS2}, Section 1.1.2). Moreover, an economic motivation for focusing on positive portfolios only is given by the fact that market participants have \emph{limited liability} and, therefore, are not allowed to trade anymore if their total tradeable wealth reaches zero. See also Section 2 of \cite{CL}, Section 6 of \cite{P2} and Section 10.3 of \cite{PH} for an amplification of the latter point.
\end{Rem}

\section{No-arbitrage conditions and the market price of risk}	\label{S1.1bis}

In order to ensure that the model introduced in the previous Section represents a viable financial market, in a sense to be made precise (see Definition \ref{arb}), we need to carefully answer the question of whether properly defined arbitrage opportunities are excluded. We start by giving the following Definition.

\begin{Def}	\label{incrp}
A trading strategy $\pi\in\mathcal{A}$ is said to yield an \emph{increasing profit} if the corresponding discounted portfolio process $\bar{V}^{\,\pi}=\left(\bar{V}_t^{\,\pi}\right)_{0\leq t \leq T}$ satisfies the following two conditions:
\begin{itemize}
\item[(a)] $\bar{V}^{\,\pi}$ is $P$-a.s. increasing, in the sense that $P\left(\bar{V}_s^{\,\pi}\leq \bar{V}_t^{\,\pi} \text{ for all } s,t\in\left[0,T\right] \text{ with } s\leq t\right)=1$;
\item[(b)] $P\left(\bar{V}_T^{\,\pi}>1\right)>0$.
\end{itemize}
\end{Def}

The notion of increasing profit represents the most glaring type of arbitrage opportunity and, hence, it is of immediate interest to know whether it is allowed or not in the financial market. As a preliminary, the following Lemma gives an equivalent characterization of the notion of increasing profit. We denote by $\ell$ the Lebesgue measure on $\left[0,T\right]$.

\begin{Lem}	\label{Lem1}
There exists an increasing profit if and only if there exists a trading strategy $\pi\in\mathcal{A}$ satisfying the following two conditions:
\begin{itemize}
\item[(a)] $\pi'_t\,\sigma_t=0$ $P\otimes\ell$-a.e.;
\item[(b)] $\pi'_t\left(\mu_t-r_t\mathbf{1}\right)\neq 0$ on some subset of $\,\Omega\times\left[0,T\right]$ with positive $P\otimes\ell$-measure.
\end{itemize}
\end{Lem}
\begin{proof}
Let $\pi\in\mathcal{A}$ be a trading strategy yielding an increasing profit. Due to Definition \ref{incrp}, the process $\bar{V}^{\,\pi}$ is $P$-a.s. increasing, hence of finite variation. Equation \eqref{E1-8} then implies that the continuous local martingale $\left(\int_0^t\bar{V}_u^{\,\pi}\,\pi_u'\,\sigma_u\,dW_u\right)_{0\leq t \leq T}$ is also of finite variation. This fact in turn implies that $\pi_t'\,\sigma_t=0$ $P\otimes\ell$-a.e. (see e.g. \cite{KS1}, Section 1.5). Since $P\left(\bar{V}_T^{\,\pi}>1\right)>0$, we must have $\pi'_t\left(\mu_t-r_t\mathbf{1}\right)\neq 0$ on some subset of $\Omega\times\left[0,T\right]$ with non-zero $P\otimes\ell$-measure.

Conversely, let $\pi\in\mathcal{A}$ be a trading strategy satisfying conditions \emph{(a)}-\emph{(b)}. Define then the process $\bar{\pi}=\left(\bar{\pi}_t\right)_{0\leq t \leq T}$ as follows, for $t\in\left[0,T\right]$:
$$
\bar{\pi}_t:=\sign\bigl(\pi_t'\left(\mu_t-r_t\mathbf{1}\right)\bigr)\pi_t
$$
It is clear that $\bar{\pi}\in\mathcal{A}$ and $\bar{\pi}_t'\,\sigma_t=0$ $P\otimes\ell$-a.e. and hence, due to \eqref{E1-5}, for all $t\in\left[0,T\right]$: 
$$
\bar{V}^{\,\bar{\pi}}_t = \exp\left(\int_0^t\bar{\pi}_u'\left(\mu_u-r_u\mathbf{1}\right)du\right)
$$
Furthermore, we have that $\bar{\pi}_t'\left(\mu_t-r_t\mathbf{1}\right)\geq 0$, with strict inequality holding on some subset of $\Omega\times\left[0,T\right]$ with non-zero $P\otimes\ell$-measure. This implies that the process $\bar{V}^{\,\bar{\pi}}=\left(\bar{V}^{\,\bar{\pi}}_t\right)_{0\leq t \leq T}$ is $P$-a.s. increasing and satisfies $P\left(\bar{V}_T^{\,\bar{\pi}}>1\right)>0$, thus showing that $\bar{\pi}$ yields an increasing profit.
\end{proof}

\begin{Rem}
According to Definition 3.9 in \cite{KK}, a trading strategy satisfying conditions \emph{(a)}-\emph{(b)} of Lemma \ref{Lem1} is said to yield an \emph{immediate arbitrage opportunity} (see \cite{DS4} and Section 4.3.2 of \cite{Fo} for a thorough analysis of the concept). In a general semimartingale setting, Proposition 3.10 of \cite{KK} extends our Lemma \ref{Lem1} and shows that the absence of (unbounded) increasing profits is equivalent to the absence of immediate arbitrage opportunities. 
\end{Rem}

The following Proposition gives a necessary and sufficient condition in order to exclude the existence of increasing profits.

\begin{Prop}	\label{Prop1}
There are no increasing profits if and only if there exists an $\R^d$-valued progressively measurable process $\gamma=\left(\gamma_t\right)_{0\leq t \leq T}$ such that the following condition holds:
\be	\label{Prop1-1}
\mu_t-r_t\mathbf{1} = \sigma_t\gamma_t
\qquad \text{$P\otimes\ell$-a.e.}
\ee
\end{Prop}
\begin{proof}
Suppose there exists an $\R^d$-valued progressively measurable process $\gamma=\left(\gamma_t\right)_{0\leq t \leq T}$ such that condition \eqref{Prop1-1} is satisfied and let $\pi\in\mathcal{A}$ be such that $\pi_t'\,\sigma_t=0$ $P\otimes\ell$-a.e. Then we have:
$$
\pi_t'\left(\mu_t-r_t\mathbf{1}\right)
=\pi_t'\,\sigma_t\gamma_t = 0
\qquad \text{$P\otimes\ell$-a.e.}
$$
meaning that there cannot exist a trading strategy $\pi\in\mathcal{A}$ satisfying conditions \emph{(a)}-\emph{(b)} of Lemma \ref{Lem1}. Due to the equivalence result of Lemma \ref{Lem1}, this implies that there are no increasing profits.

Conversely, suppose that there exists no trading strategy in $\mathcal{A}$ yielding an increasing profit. Let us first introduce the following linear spaces, for every $t\in\left[0,T\right]$:
$$
\mathcal{R}\left(\sigma_t\right):=\left\{\sigma_t y:y\in\R^d\right\}
\qquad
\mathcal{K}\left(\sigma_t'\right):=\left\{y\in\R^N:\sigma_t'y=0\right\}
$$
Denote by $\Pi_{\mathcal{K}\left(\sigma_t'\right)}$ the orthogonal projection on $\mathcal{K}\left(\sigma_t'\right)$. As in Lemma 1.4.6 of \cite{KS2}, we define the process $p=\left(p_t\right)_{0\leq t \leq T}$ by:
$$
p_t:=\Pi_{\mathcal{K}\left(\sigma_t'\right)}\left(\mu_t-r_t\mathbf{1}\right)
$$
Define then the process $\hat{\pi}=\left(\hat{\pi}_t\right)_{0\leq t \leq T}$ by:
$$
\hat{\pi}_t:=\begin{cases}
\frac{p_t}{\left\|p_t\right\|} & \text{if } p_t\neq 0,	\\
0	& \text{if } p_t=0.
\end{cases}
$$
Since the processes $\mu$ and $r$ are progressively measurable, Corollary 1.4.5 of \cite{KS2} ensures that $\hat{\pi}$ is progressively measurable. Clearly, we have then $\hat{\pi}\in\mathcal{A}$ and, by construction, $\hat{\pi}$  satisfies condition \emph{(a)} of Lemma \ref{Lem1}. Since there are no increasing profits, Lemma \ref{Lem1} implies that the following identity holds $P\otimes\ell$-a.e.:
\be	\label{normp}
\left\|p_t\right\|
=\frac{p_t'}{\left\|p_t\right\|}
\left(\mu_t-r_t\mathbf{1}\right)\ind_{\left\{p_t\neq 0\right\}}
=\hat{\pi}_t'
\left(\mu_t-r_t\mathbf{1}\right)\ind_{\left\{p_t\neq 0\right\}}=0
\ee
where the first equality uses the fact that $\mu_t-r_t\mathbf{1}-p_t\in\mathcal{K}^{\perp}\left(\sigma_t'\right)$, for all $t\in\left[0,T\right]$, with the superscript $\perp$ denoting the orthogonal complement. From \eqref{normp} we have $p_t=0$ $P\otimes\ell$-a.e., meaning that $\mu_t-r_t\mathbf{1}\in\mathcal{K}^{\perp}\left(\sigma_t'\right)=\mathcal{R}\left(\sigma_t\right)$ $P\otimes\ell$-a.e. This amounts to saying that we have:
$$
\mu_t-r_t\mathbf{1}=\sigma_t\gamma_t
\qquad \text{$P\otimes\ell$-a.e.}
$$
for some $\gamma_t\in\R^d$.
Taking care of the measurability issues, it can be shown that we can take $\gamma=\left(\gamma_t\right)_{0\leq t \leq T}$ as a progressively measurable process (compare \cite{KS2}, proof of Theorem 1.4.2).
\end{proof}

Let us now introduce one of the crucial objects in our analysis: the \emph{market price of risk} process.

\begin{Def}	\label{D1-2}
The $\R^d$-valued progressively measurable \emph{market price of risk} process $\theta=\left(\theta\right)_{0\leq t \leq T}$ is defined as follows, for $t\in\left[0,T\right]$:
$$
\theta_t:=\sigma_t'\left(\sigma_t\,\sigma_t'\right)^{-1}\left(\mu_t-r_t\mathbf{1}\right)
$$
\end{Def}

The standing Assumption \ref{A1-1} ensures that the market price of risk process $\theta$ is well-defined\footnote{It is worth pointing out that, if Assumption \ref{A1-1} does not hold but condition \eqref{Prop1-1} is satisfied, i.e. we have $\mu_t-r_t\mathbf{1}\in\mathcal{R}\left(\sigma_t\right)$ $P\otimes\ell$-a.e., then the market price of risk process $\theta$ can still be defined by replacing $\sigma_t'\left(\sigma_t\,\sigma_t'\right)^{-1}$ with the \emph{Moore-Penrose pseudoinverse} of the matrix $\sigma_t$.}. From a financial perspective, $\theta_t$ measures the excess return $\left(\mu_t-r_t\mathbf{1}\right)$ of the risky assets (with respect to the savings account) in terms of their volatility.

\begin{Rem}[\textbf{\emph{Absence of increasing profits}}]	\label{Rem-D1-2}
Note that, by definition, the market price of risk process $\theta$ satisfies condition \eqref{Prop1-1}. Proposition \ref{Prop1} then implies that, under the standing Assumption \ref{A1-1}, there are no increasing profits. Note however that $\theta$ may not be the unique process satisfying condition \eqref{Prop1-1}.
\end{Rem}

Let us now introduce the following integrability condition on the market price of risk process.

\begin{Ass}	\label{A1-2}
The market price of risk process $\theta=\left(\theta_t\right)_{0\leq t \leq T}$ belongs to $L^2_{loc}\left(W\right)$, meaning that $\int_0^T\left\|\theta_t\right\|^2dt<\infty$ $P$-a.s.
\end{Ass}

\begin{Rem}	\label{minimal}
Let $\gamma=\left(\gamma_t\right)_{0\leq t \leq T}$ be an $\R^d$-valued progressively measurable process satisfying condition \eqref{Prop1-1}. Letting $\mathcal{R}\left(\sigma_t'\right)=\left\{\sigma_t'\,x:x\in\R^N\right\}$ and $\mathcal{R}^{\perp}\left(\sigma_t'\right)=\mathcal{K}\left(\sigma_t\right)=\left\{x\in\R^d:\sigma_t\,x=0\right\}$, we get the orthogonal decomposition  $\gamma_t=\Pi_{\mathcal{R}\left(\sigma_t'\right)}\left(\gamma_t\right)+\Pi_{\mathcal{K}\left(\sigma_t\right)}\left(\gamma_t\right)$, for $t\in\left[0,T\right]$. Under Assumption \ref{A1-1}, elementary linear algebra gives that $\Pi_{\mathcal{R}\left(\sigma_t'\right)}\left(\gamma_t\right)=\sigma_t'\left(\sigma_t\sigma_t'\right)^{-1}\sigma_t\gamma_t=\sigma_t'\left(\sigma_t\sigma_t'\right)^{-1}\left(\mu_t-r_t\mathbf{1}\right)=\theta_t$, thus giving $\left\|\gamma_t\right\|=\left\|\theta_t\right\|+\left\|\Pi_{\mathcal{K}\left(\sigma_t\right)}\left(\gamma_t\right)\right\|\geq\left\|\theta_t\right\|$, for all $t\in\left[0,T\right]$. This implies that, as soon as there exists \emph{some} $\R^d$-valued progressively measurable process $\gamma$ satisfying \eqref{Prop1-1} and such that $\gamma\in L^2_{loc}\left(W\right)$, then the market price of risk process $\theta$ satisfies Assumption \ref{A1-2}. In other words, the risk premium process $\theta$ introduced in Definition \ref{D1-2} enjoys a \emph{minimality} property among all progressively measurable processes $\gamma$ which satisfy condition \eqref{Prop1-1}.
\end{Rem}

Many of our results will rely on the key relation existing between Assumption \ref{A1-2} and no-arbitrage, which has been first examined in \cite{AS1} and \cite{S1} and also plays a crucial role in \cite{DS4} and \cite{LS}.
We now introduce a fundamental local martingale associated to the market price of risk process $\theta$. Let us define the process $\widehat{Z}=\bigl(\widehat{Z}_t\bigr)_{0\leq t \leq T}$ as follows, for all $t\in\left[0,T\right]$:
\be	\label{E1-9}
\widehat{Z}_t:=\mathcal{E}\left(-\int\theta'dW\right)_t 
=\exp\left(-\sum_{j=1}^d\int_0^t\theta_u^j\,dW_u^j-\frac{1}{2}\sum_{j=1}^d\int_0^t\left(\theta_u^j\right)^2du\right)
\ee
Note that Assumption \ref{A1-2} ensures that the stochastic integral $\int\theta'dW$ is well-defined as a continuous local martingale. It is well-known that $\widehat{Z}=\bigl(\widehat{Z}_t\bigr)_{0\leq t \leq T}$ is a strictly positive continuous local martingale with $\widehat{Z}_0=1$. Due to Fatou's Lemma, the process $\widehat{Z}$ is also a supermartingale (see e.g. \cite{KS1}, Problem 1.5.19) and, hence, we have $\E\bigl[\widehat{Z}_T\bigr]\leq\E\bigl[\widehat{Z}_0\bigr]=1$. It is easy to show that the process $\widehat{Z}$ is a true martingale, and not only a local martingale, if and only if $E\bigl[\widehat{Z}_T\bigr]=E\bigl[\widehat{Z}_0\bigr]=1$. However, it may happen that the process $\widehat{Z}$ is a \emph{strict} local martingale, i.e. a local martingale which is not a true martingale.
In any case, the following Proposition shows the basic property of the process $\widehat{Z}$.
\begin{Prop}	\label{P1-1}
Suppose that Assumption \ref{A1-2} holds and let $\widehat{Z}=\bigl(\widehat{Z}_t\bigr)_{0\leq t \leq T}$ be defined as in \eqref{E1-9}. Then the following hold:
\begin{itemize}
\item[(a)] for all $i=1,\ldots,N$, the process $\widehat{Z}\,\bar{S}^i=\bigl(\widehat{Z}_t\,\bar{S}^i_t\bigr)_{0\leq t \leq T}$ is a local martingale;
\item[(b)] for any trading strategy $\pi\in\mathcal{A}$ the process $\widehat{Z}\,\bar{V}^{\,\pi}=\bigl(\widehat{Z}_t\,\bar{V}^{\,\pi}_t\bigr)_{0\leq t \leq T}$ is a local martingale.
\end{itemize} 
\end{Prop}
\begin{proof}
Part \emph{(a)} follows from part \emph{(b)} by taking $\pi\in\mathcal{A}$ with $\pi^i\equiv 1$ and $\pi^j\equiv 0$ for $j\neq i$, for any $i=1,\ldots,N$. Hence, it suffices to prove part \emph{(b)}. Recalling equation \eqref{E1-8}, an application of the product rule gives:
\be	\label{E1-10}	\ba
d\bigl(\widehat{Z}_t\,\bar{V}_t^{\,\pi}\bigr) &=
\bar{V}_t^{\,\pi}\,d\widehat{Z}_t+\widehat{Z}_t\,d\bar{V}_t^{\,\pi}+d\bigl\langle\bar{V}^{\,\pi},\widehat{Z}\bigr\rangle_t	\\
&= -\bar{V}_t^{\,\pi}\,\widehat{Z}_t\,\theta_t'\,dW_t
+\widehat{Z}_t\,\bar{V}_t^{\,\pi}\,\pi_t'\left(\mu_t-r_t\mathbf{1}\right)dt
+\widehat{Z}_t\,\bar{V}_t^{\,\pi}\,\pi_t'\,\sigma_t\,dW_t
-\widehat{Z}_t\,\bar{V}_t^{\,\pi}\,\pi_t'\,\sigma_t\,\theta_t\,dt	\\
&= \widehat{Z}_t\,\bar{V}_t^{\,\pi}\left(\pi_t'\,\sigma_t-\theta_t'\right)dW_t
\ea	\ee
Since $\sigma'\pi\in L^2_{loc}\left(W\right)$ and $\theta\in L^2_{loc}\left(W\right)$, this shows the local martingale property of $\widehat{Z}\,\bar{V}^{\,\pi}$.
\end{proof}

Under the standing Assumption \ref{A1-1}, we have seen that the diffusion-based financial market described in Section \ref{S1.1} does not allow for increasing profits (see Remark \ref{Rem-D1-2}). However, the concept of increasing profit represents an almost pathological notion of arbitrage opportunity. Hence, we would like to know whether weaker and more economically meaningful types of arbitrage opportunities can exist. To this effect, let us give the following Definition, adapted from \cite{Ka2}.

\begin{Def}	\label{arb}
An $\F$-measurable non-negative random variable $\xi$ is called an \emph{arbitrage of the first kind} if $P\left(\xi>0\right)>0$ and, for all $v\in\left(0,\infty\right)$, there exists a trading strategy $\pi^v\in\mathcal{A}$ such that $\bar{V}^{\,v,\pi^v}_T\geq\xi$ $P$-a.s.
We say that the financial market is \emph{viable} if there are no arbitrages of the first kind.
\end{Def}

The following Proposition shows that the existence of an increasing profit implies the existence of an arbitrage of the first kind. Due to the It\^o-process framework considered in this paper, we are able to provide a simple proof.

\begin{Prop}	\label{ip-arb}
Let $\pi\in\mathcal{A}$ be a trading strategy yielding an increasing profit. Then there exists an arbitrage of the first kind.
\end{Prop}
\begin{proof}
Let $\pi\in\mathcal{A}$ yield an increasing profit and define $\xi:=\bar{V}_T^{\,\pi}-1$. Due to Definition \ref{incrp}, we have $P\left(\xi\geq 0\right)=1$ and $P\left(\xi>0\right)>0$. Then, for any $v\in\left[1,\infty\right)$, we have $\bar{V}_T^{\,v,\pi}=v\bar{V}_T^{\,\pi}>v\,\xi\geq\xi$ $P$-a.s. For any $v\in\left(0,1\right)$, let us define $\pi^v_t:=-\frac{\log\left(v\right)+\log\left(1-v\right)}{v}\,\pi_t$. Clearly, for any $v\in\left(0,1\right)$, the process $\pi^v=\left(\pi^v_t\right)_{0\leq t \leq T}$ satisfies $\pi^v\in\mathcal{A}$ and, due to Lemma \ref{Lem1}, $\left(\pi^v_t\right)'\sigma_t=0$ $P\otimes\ell$-a.e. We have then:
$$
\bar{V}_T^{\,v,\pi^v}
= v\exp\left(\int_0^T\left(\pi^v_t\right)'\left(\mu_t-r_t\mathbf{1}\right)dt\right)
= v\left(\bar{V}_T^{\,\pi}\right)^{-\frac{\log\left(v\right)+\log\left(1-v\right)}{v}}
> \bar{V}_T^{\,\pi}-1 = \xi
\qquad \text{$P$-a.s.}
$$
where the second equality follows from the elementary identity $\exp\left(\alpha x\right)=\left(\exp x\right)^{\alpha}$ and the last inequality follows since $vx^{-\frac{\log\left(v\right)+\log\left(1-v\right)}{v}}>x-1$ for $x\geq 1$ and for every $v\in\left(0,1\right)$.
We have thus shown that, for every $v\in\left(0,\infty\right)$, there exists a trading strategy $\pi^v\in\mathcal{A}$ such that $\bar{V}_T^{\,v,\pi^v}\geq\xi$ $P$-a.s., meaning that the random variable $\xi=\bar{V}_T^{\,\pi}-1$ is an arbitrage of the first kind.
\end{proof}

\begin{Rem}	\label{Rem-NSA}
As we shall see by means of a simple example after Corollary \ref{NOarb}, there are instances of models where there are no increasing profits but there are arbitrages of the first kind, meaning that the absence of arbitrages of the first kind is a \emph{strictly} stronger no-arbitrage-type condition than the absence of increasing profits. Furthermore, there exists a notion of arbitrage opportunity lying between the notion of increasing profit and that of arbitrage of the first kind, namely the notion of \emph{strong arbitrage opportunity}, which consists of a trading strategy $\pi\in\mathcal{A}$ such that $\bar{V}^{\,\pi}_t\geq 1$ $P$-a.s. for all $t\in\left[0,T\right]$ and $P\left(\bar{V}^{\,\pi}_T>1\right)>0$. It can be shown that there are no strong arbitrage opportunities if and only if there are no increasing profits \emph{and} the process $\bigl(\int_0^t\left\|\theta_u\right\|^2\!du\bigr)_{0\leq t \leq T}$ does not jump to infinity on $\left[0,T\right]$. For simplicity of presentation, we omit the details and refer instead the interested reader to Theorem 3.5 of \cite{Str} (where the absence of strong arbitrage opportunities is denoted as condition NA$_+$) and Section 4.3.2 of \cite{Fo}. We want to point out that the notion of \emph{strong arbitrage opportunity} plays an important role in the context of the \emph{benchmark approach}, see e.g. Section 6 of \cite{P2}, Section 10.3 of \cite{PH} and Remark 4.3.9 of \cite{Fo}.
\end{Rem}

We now proceed with the question of whether arbitrages of the first kind are allowed in our financial market model. To this effect, let us first give the following Definition.

\begin{Def}	\label{defl}
A \emph{martingale deflator} is a real-valued non-negative adapted process $D=\left(D_t\right)_{0\leq t \leq T}$ with $D_0=1$ and $D_T>0$ $P$-a.s. and such that the process $D\bar{V}^{\,\pi}=\left(D_t\bar{V}^{\,\pi}_t\right)_{0\leq t \leq T}$ is a local martingale for every $\pi\in\mathcal{A}$. We denote by $\mathcal{D}$ the set of all martingale deflators.
\end{Def}

\begin{Rem}	\label{positive}
Let $D\in\mathcal{D}$. Then, taking $\pi\equiv 0$, Definition \ref{defl} implies that $D$ is a non-negative local martingale and hence, due to Fatou's Lemma, also a supermartingale. Since $D_T>0$ $P$-a.s., the minimum principle for non-negative supermartingales (see e.g. \cite{RY}, Proposition II.3.4) implies that $P\bigl(D_t>0, D_{t-}>0 \text{ for all }t\in\left[0,T\right]\bigr)=1$.
\end{Rem}

Note that part \emph{(b)} of Proposition \ref{P1-1} implies that, as soon as Assumption \ref{A1-2} is satisfied, the process $\widehat{Z}=\bigl(\widehat{Z}_t\bigr)_{0\leq t \leq T}$ introduced in \eqref{E1-9} is a martingale deflator, in the sense of Definition \ref{defl}. 
The following Lemma describes the general structure of martingale deflators. Related results can also be found in \cite{AS1}, \cite{AS2} and \cite{S2}.

\begin{Lem}	\label{strdefl}
Let $D=\left(D_t\right)_{0\leq t \leq T}$ be a martingale deflator. Then there exist an $\R^d$-valued progressively measurable process $\gamma=\left(\gamma_t\right)_{0\leq t \leq T}$ in $L^2_{loc}\left(W\right)$ satisfying condition \eqref{Prop1-1} and a real-valued local martingale $N=\left(N_t\right)_{0\leq t \leq T}$ with $N_0=0$, $\Delta N>-1$ $P$-a.s. and $\langle N,W^i\rangle\equiv 0$, for all $i=1,\ldots,d$, such that the following hold, for all $t\in\left[0,T\right]$:
\be	\label{strdefl-1}
D_t = \mathcal{E}\left(-\int\gamma\,dW+N\right)_t
\ee
\end{Lem}
\begin{proof}
Let us define the process $L:=\int D_-^{-1}dD$. Due to Remark \ref{positive}, the process $D_-^{-1}$ is well-defined and, being adapted and left-continuous, is also predictable and locally bounded. Since $D$ is a local martingale, this implies that the process $L$ is well-defined as a local martingale null at 0 and we have $D=\mathcal{E}\left(L\right)$. The Kunita-Watanabe decomposition (see \cite{AS3}, case 3) allows us to represent the local martingale $L$ as follows:
$$
L=-\int\gamma\,dW+N
$$
for some $\R^d$-valued progressively measurable process $\gamma=\left(\gamma_t\right)_{0\leq t \leq T}$ belonging to $L^2_{loc}\left(W\right)$, i.e. satisfying $\int_0^T\left\|\gamma_t\right\|^2dt<\infty$ $P$-a.s., and for some local martingale $N=\left(N_t\right)_{0\leq t \leq T}$ with $N_0=0$ and $\langle N,W^i\rangle\equiv 0$ for all $i=1,\ldots,d$. Furthermore, since $\left\{D>0\right\}=\left\{\Delta L>-1\right\}$ and $\Delta L=\Delta N$, we have that $\Delta N>-1$ $P$-a.s. It remains to show that $\gamma$ satisfies condition \eqref{Prop1-1}. Let $\pi\in\mathcal{A}$. Then, by using the product rule and recalling equation \eqref{E1-8}:
\be	\label{strdefl-2}	\ba
d\left(D\bar{V}^{\,\pi}\right)_t
&= D_{t-}\,d\bar{V}^{\,\pi}_t+\bar{V}^{\,\pi}_tdD_t+d\langle D,\bar{V}^{\,\pi}\rangle_t	\\
&= D_{t-}\bar{V}^{\,\pi}_t\pi'_t\left(\mu_t-r_t\mathbf{1}\right)dt+D_{t-}\bar{V}^{\,\pi}_t\pi'_t\sigma_t\,dW_t	
+\bar{V}^{\,\pi}_tD_{t-}\,dL_t+D_{t-}\bar{V}^{\,\pi}_t\,d\Bigl\langle L,\int\pi'\sigma\,dW\Bigr\rangle_t	\\
&= D_{t-}\bar{V}^{\,\pi}_t\pi'_t\left(\mu_t-r_t\mathbf{1}\right)dt+D_{t-}\bar{V}^{\,\pi}_t\pi'_t\sigma_t\,dW_t
+\bar{V}^{\,\pi}_tD_{t-}\,dL_t-D_{t-}\bar{V}^{\,\pi}_t\pi'_t\sigma_t\gamma_t\,dt	\\
&= D_{t-}\bar{V}^{\,\pi}_t\pi'_t\sigma_t\,dW_t+\bar{V}^{\,\pi}_tD_{t-}\,dL_t
+D_{t-}\bar{V}^{\,\pi}_t\pi'_t\left(\mu_t-r_t\mathbf{1}-\sigma_t\gamma_t\right)dt
\ea	\ee
Since $D\in\mathcal{D}$, the product $D\bar{V}^{\,\pi}$ is a local martingale, for every $\pi\in\mathcal{A}$. This implies that the continuous finite variation term in \eqref{strdefl-2} must vanish. Since $D_-$ and $\bar{V}^{\,\pi}$ are $P$-a.s. strictly positive and $\pi\in\mathcal{A}$ was arbitrary, this implies that condition \eqref{Prop1-1} must hold.
\end{proof}

The following Proposition shows that the existence of a martingale deflator is a sufficient condition for the absence of arbitrages of the first kind. 

\begin{Prop}	\label{Pdefl}
If $\mathcal{D}\neq\emptyset$ then there cannot exist arbitrages of the first kind.
\end{Prop}
\begin{proof}
Let $D\in\mathcal{D}$ and suppose that there exists a random variable $\xi$ yielding an arbitrage of the first kind. Then, for every $n\in\N$, there exists a strategy $\pi^n\in\mathcal{A}$ such that $\bar{V}_T^{\,1/n,\pi^n}\geq\xi$ $P$-a.s. For every $n\in\N$, the process $D\bar{V}^{\,1/n,\pi^n}=\bigl(D_t\bar{V}^{\,1/n,\pi^n}_t\bigr)_{0\leq t \leq T}$ is a positive local martingale and, hence, a supermartingale. So, for every $n\in\N$:
$$
E\left[D_T\,\xi\right]
\leq E\left[D_T\bar{V}_T^{\,1/n,\pi^n}\right]
\leq E\left[D_0\bar{V}_0^{\,1/n,\pi^n}\right]
= \frac{1}{n}
$$
Letting $n\rightarrow\infty$ gives $E\left[D_T\,\xi\right]=0$ and hence $D_T\,\xi=0$ $P$-a.s. Since, due to Definition \ref{defl}, we have $D_T>0$ $P$-a.s. this implies that $\xi=0$ $P$-a.s., which contradicts the assumption that $\xi$ is an arbitrage of the first kind.
\end{proof}

It is worth pointing out that one can also prove a converse result to Proposition \ref{Pdefl}, showing that if there are no arbitrages of the first kind then the set $\mathcal{D}$ is non-empty. In a general semimartingale setting, this has been recently shown in \cite{Ka2} (see also Section 4 of \cite{Fo} and \cite{HS} in the context of continuous path processes). Furthermore, Proposition 1 of \cite{Ka1} shows that the absence of arbitrages of the first kind is equivalent to the condition of \emph{No Unbounded Profit with Bounded Risk} (NUPBR), formally defined as the condition that the set $\left\{\bar{V}^{\,\pi}_T:\pi\in\mathcal{A}\right\}$ be bounded in probability\footnote{The (NUPBR) condition has been introduced under that name in \cite{KK}. However, the condition that the set $\left\{\bar{V}^{\,\pi}_T:\pi\in\mathcal{A}\right\}$ be bounded in probability also plays a key role in the seminal work \cite{DS} and its implications have been systematically studied in \cite{Ka}, where the same condition is denoted as ``property BK''.}. By relying on these facts, we can state the following Theorem\footnote{We want to remark that an analogous result has already been given in Theorem 2 of \cite{LW} under the assumption of a complete financial market.}, the second part of which follows from Proposition 4.19 of \cite{KK}.

\begin{Thm}	\label{ThmKK}
The following are equivalent:
\begin{itemize}
\item[(a)] $\mathcal{D}\neq\emptyset$;
\item[(b)] there are no arbitrages of the first kind;
\item[(c)] $\left\{\bar{V}^{\,\pi}_T:\pi\in\mathcal{A}\right\}$ is bounded in probability, i.e. the (NUPBR) condition holds.
\end{itemize}
Moreover, for every concave and strictly increasing utility function $U:\left[0,\infty\right)\rightarrow\R$, the expected utility maximisation problem of finding an element $\pi^*\in\mathcal{A}$ such that
$$
E\left[U\left(\bar{V}_T^{\,\pi^*}\right)\right] =
\underset{\pi\in\mathcal{A}}{\sup}\,E\left[U\left(\bar{V}_T^{\,\pi}\right)\right]
$$
either does not have a solution or has infinitely many solutions when any of the conditions (a)-(c) fails.
\end{Thm}

In view of the second part of the above Theorem, the condition of absence of arbitrages of the first kind can be seen as the minimal no-arbitrage condition in order to be able to meaningfully solve portfolio optimisation problems. 

\begin{Rem}
We have defined the notion of \emph{viability} for a financial market in terms of the absence of arbitrages of the first kind (see Definition \ref{arb}). In \cite{LW}, a financial market is said to be viable if any agent with sufficiently regular preferences and with a positive initial endowment can construct an optimal portfolio. The last part of Theorem \ref{ThmKK} gives a correspondence between these two notions of viability, since it shows that the absence of arbitrages of the first kind is the minimal no-arbitrage-type condition in order to being able to meaningfully solve portfolio optimisation problems.
\end{Rem}

It is now straightforward to show that, as soon as Assumption \ref{A1-2} holds, the diffusion-based model introduced in Section \ref{S1.1} satisfies the equivalent conditions of Theorem \ref{ThmKK}. In fact, due to Proposition \ref{P1-1}, the process $\widehat{Z}$ defined in \eqref{E1-9} is a martingale deflator for the financial market $\left(S^0,S^1,\ldots,S^N\right)$ as soon as Assumption \ref{A1-2} is satisfied and, hence, due to Proposition \ref{Pdefl}, there are no arbitrages of the first kind. Conversely, suppose that there are no arbitrages of the first kind but Assumption \ref{A1-2} fails to hold. Then, due to Remark \ref{minimal} together with Lemma \ref{strdefl}, we have that $\mathcal{D}=\emptyset$. Theorem \ref{ThmKK} then implies that there exist arbitrages of the first kind, thus leading to a contradiction.
We have thus proved the following Corollary.

\begin{Cor}	\label{NOarb}
The financial market $\left(S^0,S^1,\ldots,S^N\right)$ is viable, i.e. it does not admit arbitrages of the first kind (see Definition \ref{arb}), if and only if Assumption \ref{A1-2} holds.
\end{Cor}

As we have seen in Proposition \ref{ip-arb}, if there exist an increasing profit then there exist an arbitrage of the first kind. We now show that the absence of arbitrages of the first kind is a \emph{strictly} stronger no-arbitrage-type condition than the absence of increasing profits by means of a simple example, which we adapt from Example 3.4 of \cite{DS4}. Let $N=d=1$, $r\equiv 0$ and let the real-valued process $S=\left(S_t\right)_{0\leq t \leq T}$ be given as the solution to the following SDE:
$$
dS_t = \frac{S_t}{\sqrt{t}}\,dt+S_t\,dW_t
\qquad S_0=s\in\left(0,\infty\right)
$$
Using the notations introduced in Section \ref{S1.1}, we have $\mu_t=1/\sqrt{t}$, for $t\in\left[0,T\right]$, and $\sigma\equiv 1$. Clearly, condition \eqref{Prop1-1} is satisfied, since we trivially have $\mu_t=\sigma_t\theta_t$, where $\theta_t=1/\sqrt{t}$, for $t\in\left[0,T\right]$. Proposition \ref{Prop1} then implies that there are no increasing profits. However, $\theta\notin L^2_{loc}\left(W\right)$, since $\int_0^t\theta^2_u\,du=\int_0^t\frac{1}{u}\,du=\infty$ for all $t\in\left[0,T\right]$. Corollary \ref{NOarb} then implies that there exist arbitrages of the first kind\footnote{More precisely, note that the process $\bigl(\int_0^t\theta_u^2du\bigr)_{0\leq t \leq T}=\bigl(\int_0^t\frac{1}{u}du\bigr)_{0\leq t \leq T}$ jumps to infinity instantaneously at $t=0$. Hence, as explained in Remark \ref{Rem-NSA}, the model considered in the present example allows not only for arbitrages of the first kind, but also for strong arbitrage opportunities. Of course, there are instances where strong arbitrage opportunities are precluded but still there exist arbitrages of the first kind. We refer the interested reader to \cite{BT} for an example of such a model, where the price of a risky asset is modelled as the exponential of a \emph{Brownian bridge} (see also \cite{LW}, example 3.1)}.

We want to emphasise that, due to Theorem \ref{ThmKK}, the diffusion-based model introduced in Section \ref{S1.1} allows us to meaningfully consider portfolio optimisation problems as soon as Assumption \ref{A1-2} holds. However, nothing guarantees that an \emph{Equivalent Local Martingale Measure} (ELMM) exists, as shown in the following classical example, already considered in \cite{DS3}, \cite{H} and \cite{KK}. Other instances of models for which an ELMM does not exist arise in the context of \emph{diverse} financial markets, see Chapter II of \cite{FK}.

\begin{Ex}
Let us suppose that $\FF=\FF^W$, where $W$ is a standard Brownian motion ($d=1$), and let $N=1$. Assume that $S_t^0\equiv 1$ for all $t\in\left[0,T\right]$ and that the real-valued process $S=\left(S_t\right)_{0\leq t \leq T}$ is given by the solution to the following SDE:
\be	\label{E3-1}
dS_t = \frac{1}{S_t}\,dt+dW_t \qquad S_0=s\in\left(0,\infty\right)
\ee
It is well-known that the process $S$ is a Bessel process of dimension three (see e.g. \cite{RY}, Section XI.1). So, $S_t$ is $P$-a.s. strictly positive and finite valued for all $t\in\left[0,T\right]$. Furthermore, the market price of risk process $\theta$ is given by $\theta_t=\sigma_t^{-1}\,\mu_t=\frac{1}{S_t}$, for $t\in\left[0,T\right]$. Since $S$ is continuous, we clearly have $\int_0^T\theta_t^2dt<\infty$ $P$-a.s., meaning that Assumption \ref{A1-2} is satisfied. Hence, due to Corollary \ref{NOarb}, there are no arbitrages of the first kind. 

However, for this particular financial market model there exists no ELMM. We prove this claim arguing by contradiction.
Suppose that $Q$ is an ELMM for $S$ and denote by $Z^Q=\bigl(Z_t^Q\bigr)_{0\leq t \leq T}$ its density process. Then, due to the martingale representation theorem (see \cite{KS1}, Theorem 3.4.15 and Problem 3.4.16), we can represent $Z^Q$ as follows:
$$
Z_t^Q=\mathcal{E}\left(-\int\lambda\,dW\right)_t
\qquad \text{ for } t\in\left[0,T\right]
$$
where $\lambda=\left(\lambda_t\right)_{0\leq t \leq T}$ is a progressively measurable process such that $\int_0^T\lambda_t^2\,dt<\infty$ $P$-a.s.
Due to Girsanov's theorem, the process $W^Q=\bigl(W_t^Q\bigr)_{0\leq t \leq T}$ defined by $W^Q_t:=W_t+\int_0^t\lambda_u\,du$, for $t\in\left[0,T\right]$, is a Brownian motion under $Q$. Hence, the process $S$ satisfies the following SDE under $Q$:
\be	\label{P3-2-2}
dS_t = \left(\frac{1}{S_t}-\lambda_t\right)dt+dW^Q_t	\qquad S_0=s
\ee
Since $Q$ is an ELMM for $S$, the SDE \eqref{P3-2-2} must have a zero drift term, i.e. it must be $\lambda_t=\frac{1}{S_t}=\theta_t$ for all $t\in\left[0,T\right]$. Then, a simple application of It\^o's formula gives:
$$
Z^Q_t = \mathcal{E}\left(-\int\frac{1}{S}\,dW\right)_t
= \exp\left(-\int_0^t\frac{1}{S_u}\,dW_u-\frac{1}{2}\int_0^t\frac{1}{S^2_u}\,du\right)
= \frac{1}{S_t}
$$
However, since $S$ is a Bessel process of dimension three, it is well-known that the process $1/S=\left(1/S_t\right)_{0\leq t \leq T}$ is a strict local martingale, i.e. it is a local martingale but not a true martingale (see e.g. \cite{RY}, Exercise XI.1.16). Clearly, this contradicts the fact that $Q$ is a well-defined probability measure\footnote{Alternatively, one can show that the probability measures $Q$ and $P$ fail to be equivalent by arguing as follows. Let us define the stopping time $\tau:=\inf\left\{t\in\left[0,T\right]:S_t=0\right\}$. The process $S=\left(S_t\right)_{0\leq t\leq T}$ is a Bessel process of dimension three under $P$ and, hence, we have $P\left(\tau<\infty\right)=0$. However, since the process $S=\left(S_t\right)_{0\leq t \leq T}$ is a $Q$-Brownian motion, we clearly have $Q\left(\tau<\infty\right)>0$. This contradicts the assumption that $Q$ and $P$ are equivalent.}, thus showing that there cannot exist an ELMM for $S$.
\end{Ex}

As the above Example shows, Assumption \ref{A1-2} does not guarantee the existence of an ELMM for the financial market $\left(S^0,S^1,\ldots,S^N\right)$. It is well-known that, in the case of continuous-path processes, the existence of an ELMM is equivalent to the \emph{No Free Lunch with Vanishing Risk} (NFLVR) no-arbitrage-type condition, see \cite{DS} and \cite{DS2}. Furthermore, (NFLVR) holds if and only if both (NUPBR) and the classical \emph{no-arbitrage} (NA) condition hold (see Section 3 of \cite{DS}, Lemma 2.2 of \cite{Ka} and Proposition 4.2 of \cite{KK}), where, recalling that $\bar{V}^{\,\pi}_0=1$, the (NA) condition precludes the existence of a trading strategy $\pi\in\mathcal{A}$ such that $P\left(\bar{V}^{\,\pi}_T\geq 1\right)=1$ and $P\left(\bar{V}_T^{\,\pi}>1\right)>0$. This implies that, even if Assumption \ref{A1-2} holds, the classical (NFLVR) condition may fail to hold. However, due to Theorem \ref{ThmKK}, the financial market may still be viable.

\begin{Rem}[\textbf{\emph{On the martingale property of $\boldsymbol{\widehat{Z}}$}}]	\label{Rem-S1.1bis}
It is important to note that Assumption \ref{A1-2} does not suffice to ensure that $\widehat{Z}$ is a true martingale. Well-known sufficient conditions for this to hold include the Novikov and Kazamaki criteria, see e.g. \cite{RY}, Section VIII.1. If $\widehat{Z}$ is a true martingale we have then $\E\bigl[\widehat{Z}_T\bigr]=1$ and we can define a probability measure $\widehat{Q}\sim P$ by letting $\frac{d\widehat{Q}}{dP}:=\widehat{Z}_T$. The martingale $\widehat{Z}$ represents then the \emph{density process} of $\widehat{Q}$ with respect to $P$, i.e. $\widehat{Z}_t=\E\Bigl[\frac{d\widehat{Q}}{dP}\bigl|\F_t\Bigr]$ $P$-a.s. for all $t\in\left[0,T\right]$, and a process $M=\left(M_t\right)_{0\leq t \leq T}$ is a local $\widehat{Q}$-martingale if and only if the process $\widehat{Z}M=\bigl(\widehat{Z}_tM_t\bigr)_{0\leq t \leq T}$ is a local $P$-martingale. Due to Proposition \ref{P1-1}-\emph{(a)}, this implies that if $\E\bigl[\widehat{Z}_T\bigr]=1$ then the process $\bar{S}:=\left(\bar{S}^1,\ldots,\bar{S}^N\right)'$ is a local $\widehat{Q}$-martingale or, in other words, the probability measure $\widehat{Q}$ is an ELMM. 
Girsanov's theorem then implies that the process $\widehat{W}=\bigl(\widehat{W}_t\bigr)_{0\leq t \leq T}$ defined by $\widehat{W}_t:=W_t+\int_0^t\theta_u\,du$ for $t\in\left[0,T\right]$ is a Brownian motion under $\widehat{Q}$. Since the dynamics of $S:=\left(S^1,\ldots,S^N\right)'$ in \eqref{E1-2} can be rewritten as:
$$
dS_t = \diag\left(S_t\right)\mathbf{1}\,r_t\,dt+\diag\left(S_t\right)\sigma_t\bigl(\theta_t\,dt+dW_t\bigr)
\qquad S_0=s 
$$
the process $\bar{S}:=\left(\bar{S}^1,\ldots,\bar{S}^N\right)'$ satisfies the following SDE under the measure $\widehat{Q}$:
$$
d\bar{S}_t 
=\diag\left(\bar{S}_t\right)\sigma_t\,d\widehat{W}_t	\qquad \bar{S}_0=s
$$
We want to point out that the process $\widehat{Z}=\bigl(\widehat{Z}_t\bigr)_{0\leq t \leq T}$ represents the density process with respect to $P$ of the \emph{minimal martingale measure}, when the latter exists, see e.g. \cite{HS}. Again, we emphasise that in this paper we do not assume neither that $\E\bigl[\widehat{Z}_T\bigr]=1$ nor that an ELMM exists.
\end{Rem}

We close this Section with a simple technical result which turns out to be useful in the following.

\begin{Lem}	\label{L1-1}
Suppose that Assumption \ref{A1-2} holds. Then an $\R^N$-valued progressively measurable process $\pi=\left(\pi_t\right)_{0\leq t \leq T}$ belongs to $\mathcal{A}$ if and only if $\int_0^T\left\|\sigma'_t\,\pi_t\right\|^2dt<\infty$ $P$-a.s.
\end{Lem}
\begin{proof}
We only need to show that Assumption \ref{A1-2} and $\int_0^T\left\|\sigma'_t\,\pi_t\right\|^2dt<\infty$ $P$-a.s. together imply that $\int_0^T\left|\pi_t'\left(\mu_t-r_t\mathbf{1}\right)\right|dt<\infty$ $P$-a.s. This follows easily from the Cauchy-Schwarz inequality, in fact:
$$
\int_0^T\left|\pi_t'\left(\mu_t-r_t\mathbf{1}\right)\right|dt
=\int_0^T\left|\pi_t'\,\sigma_t\,\theta_t\right|dt	
\leq \left(\int_0^T\left\|\sigma_t'\,\pi_t\right\|^2dt\right)^{\frac{1}{2}}
\left(\int_0^T\left\|\theta_t\right\|^2dt\right)^{\frac{1}{2}}<\infty
\quad \text{$P$-a.s.}
$$
\end{proof}

\section{The growth-optimal portfolio and the numeraire portfolio}	\label{S2}

As we have seen in the last Section, the diffusion-based model introduced in Section \ref{S1.1} can represent a viable financial market even if the traditional (NFLVR) no-arbitrage-type condition fails to hold or, equivalently, if an ELMM for $\left(S^0,S^1,\ldots,S^N\right)$ fails to exist. Let us now consider an interesting portfolio optimisation problem, namely the problem of maximising the \emph{growth rate}, formally defined as follows (compare \cite{FK}, \cite{P1} and \cite{PH}, Section 10.2).

\begin{Def}	\label{D2-1}
For a trading strategy $\pi\in\mathcal{A}$ the \emph{growth rate} process $g^{\pi}=\left(g_t^{\pi}\right)_{0\leq t \leq T}$ is defined as the drift term in the SDE satisfied by the process $\log V^{\,\pi}=\left(\,\log V_t^{\,\pi}\right)_{0\leq t \leq T}$, i.e. the term $g_t^{\pi}$ in the SDE:
\be	\label{E2-1}
d\log V^{\,\pi}_t = g_t^{\pi}\,dt+\pi'_t\sigma_t\,dW_t
\ee
A trading strategy $\pi^*\in\mathcal{A}$ (and the corresponding portfolio process $V^{\,\pi^*}$) is said to be \emph{growth-optimal} if $g_t^{\pi^*}\geq g_t^{\pi}$ $P$-a.s. for all $t\in\left[0,T\right]$ for any trading strategy $\pi\in\mathcal{A}$.
\end{Def}

The terminology ``growth rate'' is motivated by the fact that:
$$
\underset{T\rightarrow\infty}{\lim}\,\frac{1}{T}\left(\log V_T^{\,\pi}-\int_0^Tg_t^{\pi}dt\right)=0
\qquad \text{$P$-a.s.}
$$
under ``controlled growth'' of $a:=\sigma\sigma'$, i.e. $\underset{T\rightarrow\infty}{\lim}\left(\frac{\log\log T}{T^2}\int_0^Ta_t^{i,i}dt\right)=0$ $P$-a.s. (see \cite{FK}, Section 1). In the context of the general diffusion-based financial market described in Section \ref{S1.1}, the following Theorem gives an explicit description of the growth-optimal strategy $\pi^*\in\mathcal{A}$.

\begin{Thm}	\label{T2-1}
Suppose that Assumption \ref{A1-2} holds. Then there exists an unique growth-optimal strategy $\pi^*\in\mathcal{A}$, explicitly given by:
\be	\label{E2-2}
\pi_t^*=\left(\sigma_t\,\sigma_t'\right)^{-1}\sigma_t\,\theta_t
\ee
where the process $\theta=\left(\theta_t\right)_{0\leq t \leq T}$ is the market price of risk introduced in Definition \ref{D1-2}. 
The corresponding Growth-Optimal Portfolio (GOP) $V^{\,\pi^*}=\bigl(V_t^{\,\pi^*}\bigr)_{0\leq t \leq T}$ satisfies the following dynamics:
\be	\label{E2-3}
\frac{dV_t^{\,\pi^*}}{V_t^{\,\pi^*}}=r_t\,dt+\theta_t'\left(\theta_t\,dt+dW_t\right)
\ee
\end{Thm}
\begin{proof}
Let $\pi\in\mathcal{A}$ be a trading strategy. A simple application of It\^o's formula gives that:
\be	\label{T2-1-1}
d\log V_t^{\,\pi} = g_t^{\pi}\,dt+\pi_t'\,\sigma_t\,dW_t
\ee
where $g_t^{\pi}:=r_t+\pi_t'\left(\mu_t-r_t\mathbf{1}\right)-\frac{1}{2}\,\pi_t'\,\sigma_t\,\sigma_t'\,\pi_t$, for $t\in\left[0,T\right]$. Since the matrix $\sigma_t\sigma_t'$ is $P$-a.s. positive definite for all $t\in\left[0,T\right]$, due to Assumption \ref{A1-1}, a trading strategy $\pi^*\in\mathcal{A}$ is growth-optimal (in the sense of Definition \ref{D2-1}) if and only if, for every $t\in\left[0,T\right]$, $\pi^*_t$ solves the first order condition obtained by differentiating $g^{\pi}_t$ with respect to $\pi_t$. This means that $\pi^*_t$ must satisfy the following condition, for every $t\in\left[0,T\right]$:
$$
\mu_t-r_t\mathbf{1}-\sigma_t\sigma_t'\pi_t^*=0
$$
Due to Assumption \ref{A1-1}, the matrix $\sigma_t\sigma'_t$ is $P$-a.s. invertible for all $t\in\left[0,T\right]$. So, using Definition \ref{D1-2}, we get the following unique optimiser $\pi^*_t$:
$$
\pi_t^*=\left(\sigma_t\,\sigma_t'\right)^{-1}\left(\mu_t-r_t\mathbf{1}\right)
=\left(\sigma_t\,\sigma_t'\right)^{-1}\sigma_t\,\theta_t
\qquad \text{ for } t\in\left[0,T\right]
$$
We now need to verify that $\pi^*=\left(\pi_t^*\right)_{0\leq t \leq T}\in\mathcal{A}$. Due to Lemma \ref{L1-1}, it suffices to check that $\int_0^T\left\|\sigma_t'\pi_t^*\right\|^2dt<\infty$ $P$-a.s. To show this, it is enough to notice that:
$$
\int_0^T\left\|\sigma_t'\,\pi_t^*\right\|^2dt
=\int_0^T\left(\mu_t-r_t\mathbf{1}\right)'\left(\sigma_t\,\sigma_t'\right)^{-1}\left(\mu_t-r_t\mathbf{1}\right)dt
=\int_0^T\left\|\theta_t\right\|^2dt<\infty
\qquad \text{$P$-a.s.}
$$
due to Assumption \ref{A1-2}. We have thus shown that $\pi^*$ maximises the growth rate and is an admissible trading strategy. 
Finally, note that equation \eqref{T2-1-1} leads to:
$$	\ba
d\log V_t^{\,\pi^*} &= g_t^{\pi^*}\,dt+\left(\pi_t^*\right)'\sigma_t\,dW_t	\\
&= r_t\,dt+\theta_t'\,\sigma_t'\left(\sigma_t\,\sigma_t'\right)^{-1}\left(\mu_t-r_t\mathbf{1}\right)dt
-\frac{1}{2}\,\theta_t'\,\sigma_t'\left(\sigma_t\,\sigma_t'\right)^{-1}\sigma_t\,\sigma_t'\left(\sigma_t\,\sigma_t'\right)^{-1}\sigma_t\,\theta_t\,dt	\\
&\phantom{=}+\theta_t'\,\sigma_t'\left(\sigma_t\,\sigma_t'\right)^{-1}\sigma_t\,dW_t	\\
&= \Bigl(r_t+\frac{1}{2}\,\left\|\theta_t\right\|^2\Bigr)dt+\theta_t'dW_t
\ea	$$
where the last equality is obtained by replacing $\theta_t$ with its expression as given in Definition \ref{D1-2}. Equation \eqref{E2-3} then follows by a simple application of It\^o's formula.
\end{proof}

\begin{Rem} \label{Rem-GOP} $ $
\begin{enumerate}
\item
Results analogous to Theorem \ref{T2-1} can be found in Section 2 of \cite{RG}, Example 3.7.9 of \cite{KS2}, Section 2.7 of \cite{P0}, Section 3.2 of \cite{P1}, Section 10.2 of \cite{PH} and Proposition 2 of \cite{PR2}. However, in all these works the growth-optimal strategy has been derived for the specific case of a complete financial market, i.e. under the additional assumptions that $d=N$ and $\FF=\FF^W$ (see Section \ref{S1.2}). Here, we have instead chosen to deal with the more general situation described in Section \ref{S1.1}, i.e. with a general incomplete market. Furthermore, we rigorously check the admissibility of the candidate growth-optimal strategy.
\item
Due to Corollary \ref{NOarb}, Assumption \ref{A1-2} is equivalent to the absence of arbitrages of the first kind. However, it is worth emphasising that Theorem \ref{T2-1} does not rely on the existence of an ELMM for the financial market $\left(S^0,S^1,\ldots,S^N\right)$. 
\item
Due to equation \eqref{E2-3}, the discounted GOP process $\bar{V}^{\,\pi^*}=\left(\bar{V}_t^{\,\pi^*}\right)_{0\leq t \leq T}$ satisfies the following dynamics:
\be	\label{E2-4}
\frac{d\bar{V}_t^{\,\pi^*}}{\bar{V}_t^{\,\pi^*}} = \left\|\theta_t\right\|^2dt+\theta_t'\,dW_t
\ee
We can immediately observe that the drift coefficient is the ``square'' of the diffusion coefficient, thus showing that there is a strong link between instantaneous rate of return and volatility in the GOP dynamics. Moreover, the market price of risk plays a key role in the GOP dynamics (to this effect, compare the discussion in \cite{PH}, Chapter 13). Observe also that Assumption \ref{A1-2} is equivalent to requiring that the solution $\bar{V}^{\,\pi^*}$ to the SDE \eqref{E2-4} is well-defined and $P$-a.s. finite  valued, meaning that the discounted GOP does not explode in the finite time interval $\left[0,T\right]$. Indeed, it can be shown, and this holds true in general semimartingale models, that the existence of a non-explosive GOP is in fact \emph{equivalent} to the absence of arbitrages of the first kind, as can be deduced by combining Theorem \ref{ThmKK} and \cite{KK}, Theorem 4.12 (see also \cite{CL} and \cite{HS}). 
\end{enumerate}
\end{Rem}

\begin{Ex}[\emph{\textbf{The classical Black-Scholes model}}]
In order to develop an intuitive feeling for some of the concepts introduced in this Section, let us briefly consider the case of the classical Black-Scholes model, i.e. a financial market represented by $\left(S^0,S\right)$, with $r_t\equiv r$ for some $r\in\R$ for all $t\in\left[0,T\right]$ and $S=\left(S_t\right)_{0\leq t \leq T}$ a real-valued process satisfying the following SDE:
$$
dS_t = S_t\,\mu\,dt+S_t\,\sigma\,dW_t
\qquad S_0 = s\in\left(0,\infty\right)
$$
with $\mu\in\R$ and $\sigma\in\R\setminus\left\{0\right\}$. The market price of risk process $\theta=\left(\theta_t\right)_{0\leq t \leq T}$ is then given by $\theta_t\equiv\theta:=\frac{\mu-r}{\sigma}$ for all $t\in\left[0,T\right]$. Due to Theorem \ref{T2-1}, the GOP strategy $\pi^*=\left(\pi_t^*\right)_{0\leq t \leq T}$ is then given by $\pi_t^*\equiv\pi^*:=\frac{\mu-r}{\sigma^2}$, for all $t\in\left[0,T\right]$. In this special case, Novikov's condition implies that $\widehat{Z}$ is a true martingale, yielding the density process of the (minimal) martingale measure $\widehat{Q}$ (see Remark \ref{Rem-S1.1bis}).
\end{Ex}

The remaining part of this Section is devoted to the derivation of some basic but fundamental properties of the GOP. Let us start with the following simple Proposition.

\begin{Prop}	\label{P2-1}
Suppose that Assumption \ref{A1-2} holds. Then the discounted GOP process $\bar{V}^{\,\pi^*}=\bigl(\bar{V}^{\,\pi^*}_t\bigr)_{0\leq t \leq T}$ is related to the martingale deflator $\widehat{Z}=\bigl(\widehat{Z}_t\bigr)_{0\leq t \leq T}$ as follows, for all $t\in\left[0,T\right]$:
$$
\bar{V}_t^{\,\pi^*} = \frac{1}{\widehat{Z}_t}
$$
\end{Prop}
\begin{proof}
Assumption \ref{A1-2} ensures that the process $\widehat{Z}=\bigl(\widehat{Z}_t\bigr)_{0\leq t \leq T}$ is $P$-a.s. strictly positive and well-defined as a martingale deflator. Furthermore, due to Theorem \ref{T2-1}, the growth-optimal strategy $\pi^*\in\mathcal{A}$ exists and is explicitly given by \eqref{E2-2}. 
Now it suffices to observe that, due to equations \eqref{E2-4} and \eqref{E1-9}:
$$
\bar{V}_t^{\,\pi^*}
=\exp\left(\int_0^t\theta_u'dW_u+\frac{1}{2}\int_0^t\left\|\theta_u\right\|^2du\right)
=\frac{1}{\widehat{Z}_t}
$$
\end{proof}

We then immediately obtain the following Corollary.
\begin{Cor}	\label{C2-1}
Suppose that Assumption \ref{A1-2} holds. Then, for any trading strategy $\pi\in\mathcal{A}$, the process $\hat{V}^{\,\pi}=\bigl(\hat{V}^{\,\pi}_t\bigr)_{0\leq t \leq T}$ defined by $\hat{V}^{\,\pi}_t:=V_t^{\,\pi}/V_t^{\,\pi^*}$, for $t\in\left[0,T\right]$, is a non-negative local martingale and, hence, a supermartingale.
\end{Cor}
\begin{proof}
Passing to discounted quantities, we have $\hat{V}^{\pi}_t=V_t^{\,\pi}/V_t^{\,\pi^*}=\bar{V}_t^{\,\pi}/\bar{V}_t^{\,\pi^*}$. The claim then follows by combining Proposition \ref{P2-1} with part \emph{(b)} of Proposition \ref{P1-1}.
\end{proof}

In order to give a better interpretation to the preceding Corollary, let us give the following Definition, which we adapt from \cite{Be}, \cite{KK} and \cite{P2}.

\begin{Def}	\label{D2-2}
An admissible portfolio process $V^{\,\tilde{\pi}}=\bigl(V^{\,\tilde{\pi}}_t\bigr)_{0\leq t \leq T}$ has the \emph{numeraire property} if all admissible portfolio processes $V^{\,\pi}=\left(V_t^{\,\pi}\right)_{0\leq t \leq T}$ , when denominated in units of $V^{\,\tilde{\pi}}$, are supermartingales, i.e. if the process $V^{\,\pi}/V^{\,\tilde{\pi}}=\left(V^{\,\pi}_t/V_t^{\,\tilde{\pi}}\right)_{0\leq t \leq T}$ is a supermartingale for all $\pi\in\mathcal{A}$.
\end{Def}

The following Proposition shows that if a numeraire portfolio exists then it is also unique.

\begin{Prop}	\label{P2-2}
The numeraire portfolio process $V^{\,\tilde{\pi}}=\bigl(V^{\,\tilde{\pi}}_t\bigr)_{0\leq t \leq T}$ is unique (in the sense of indistinguishability). Furthermore, there exists an unique trading strategy $\tilde{\pi}\in\mathcal{A}$ such that $V^{\,\tilde{\pi}}$ is the numeraire portfolio, up to a null subset of $\Omega\times\left[0,T\right]$.
\end{Prop}
\begin{proof}
Let us first prove that if $M=\left(M_t\right)_{0\leq t \leq T}$ is a $P$-a.s. strictly positive supermartingale such that $\frac{1}{M}$ is also a supermartingale then $M_t=M_0$ $P$-a.s. for all $t\in\left[0,T\right]$. 
In fact, for any $0\leq s \leq t \leq T$:
$$
1=\frac{M_s}{M_s} \geq \frac{1}{M_s}\E\left[M_t|\F_s\right]
\geq \E\left[\frac{1}{M_t}\Bigl|\F_s\right]\E\left[M_t|\F_s\right]
\geq \frac{1}{\E\left[M_t|\F_s\right]}\E\left[M_t|\F_s\right] = 1
\qquad \text{$P$-a.s.}
$$ 
where the first inequality follows from the supermartingale property of $M$, the second from the supermartingale property of $\frac{1}{M}$ and the third from Jensen's inequality. Hence, both $M$ and $\frac{1}{M}$ are martingales. Furthermore, since we have $\E\left[\frac{1}{M_t}\bigl|\F_s\right]=\frac{1}{\E\left[M_t|\F_s\right]}$ and the function $x\mapsto x^{-1}$ is strictly convex on $\left(0,\infty\right)$, again Jensen's inequality implies that $M_t$ is $\F_s$-measurable, for all $0\leq s \leq t \leq T$. For $s=0$, this implies that $M_t=\E\left[M_t|\F_0\right]=M_0$ $P$-a.s. for all $t\in\left[0,T\right]$.

Suppose now there exist two elements $\tilde{\pi}^1,\tilde{\pi}^2\in\mathcal{A}$ such that both $V^{\,\tilde{\pi}^1}$ and $V^{\,\tilde{\pi}^2}$ have the numeraire property. By Definition \ref{D2-2}, both $V^{\,\tilde{\pi}^1}/V^{\,\tilde{\pi}^2}$ and $V^{\,\tilde{\pi}^2}/V^{\,\tilde{\pi}^1}$ are $P$-a.s. strictly positive supermartingales. Hence, it must be $V^{\,\tilde{\pi}^1}_t=V^{\,\tilde{\pi}^2}_t$ $P$-a.s. for all $t\in\left[0,T\right]$, due to the general result just proved, and thus $V^{\,\tilde{\pi}^1}$ and $V^{\,\tilde{\pi}^2}$ are indistinguishable (see \cite{KS1}, Section 1.1).
In order to show that the two trading strategies $\tilde{\pi}^1$ and $\tilde{\pi}^2$ coincide, let us write as follows:
$$	\ba
&\E\left[\int_0^T\bigl(\bar{V}_t^{\,\tilde{\pi}^1}\tilde{\pi}_t^1-\bar{V}_t^{\,\tilde{\pi}^2}\tilde{\pi}_t^2\bigr)'\sigma_t\,\sigma_t'\bigl(\bar{V}_t^{\,\tilde{\pi}^1}\tilde{\pi}_t^1-\bar{V}_t^{\,\tilde{\pi}^2}\tilde{\pi}_t^2\bigr)dt\right]	\\
&=\E\left[\biggl\langle\int\bar{V}^{\,\tilde{\pi}^1}(\tilde{\pi}^1)'\sigma\,dW
-\int\bar{V}^{\,\tilde{\pi}^2}(\tilde{\pi}^2)'\sigma\,dW\biggr\rangle_T\right]
= \E\left[\bigl\langle \bar{V}^{\,\tilde{\pi}^1}-\bar{V}^{\,\tilde{\pi}^2}\bigr\rangle_T\right] = 0
\ea	$$
where we have used equation \eqref{E1-8} and the fact that $\bar{V}^{\,\tilde{\pi}^1}$ and $\bar{V}^{\,\tilde{\pi}^2}$ are indistinguishable. Since, due to the standing Assumption \ref{A1-1}, the matrix $\sigma_t\sigma'_t$ is $P$-a.s. positive definite for all $t\in\left[0,T\right]$ and $\bar{V}^{\,\tilde{\pi}^1}$ and $\bar{V}^{\,\tilde{\pi}^2}$ are indistinguishable, this implies that it must be $\tilde{\pi}_t:=\tilde{\pi}_t^1=\tilde{\pi}_t^2$ $P\otimes\ell$-a.e., thus showing the uniqueness of the strategy $\tilde{\pi}\in\mathcal{A}$.
\end{proof}

\begin{Rem}
Note that the first part of Proposition \ref{P2-2} does not rely on any modelling assumption and, hence, is valid in full generality for any semimartingale model (compare also \cite{Be}, Section 4). 
\end{Rem}

The following fundamental Corollary makes precise the relation between the GOP, the numeraire portfolio and the viability of the financial market.

\begin{Cor}	\label{GO-num}
The financial market is viable, in the sense of Definition \ref{arb}, if and only if the numeraire portfolio exists. Furthermore, if Assumption \ref{A1-2} holds, then the growth-optimal portfolio $V^{\,\pi^*}$ coincides with the numeraire portfolio $V^{\,\tilde{\pi}}$ and the corresponding trading strategies $\pi^*, \tilde{\pi}\in\mathcal{A}$ coincide, up to a null subset of $\Omega\times\left[0,T\right]$.
\end{Cor}
\begin{proof}
If the financial market is viable, Corollary \ref{NOarb} implies that Assumption \ref{A1-2} is satisfied. Hence, due to Theorem \ref{T2-1} together with Corollary \ref{C2-1} and Definition \ref{D2-2}, the GOP exists and possesses the numeraire property. Conversely, suppose that the numeraire portfolio $V^{\,\tilde{\pi}}$ exists. Then, due to Definition \ref{D2-2}, the process $V^{\,\pi}/V^{\,\tilde{\pi}}=\bigl(V^{\,\pi}_t/V^{\,\tilde{\pi}}_t\bigr)_{0\leq t \leq T}$ is a supermartingale, for every $\pi\in\mathcal{A}$. In turn, this implies that $E\left[\bar{V}^{\,\pi}_T/\bar{V}^{\,\tilde{\pi}}_T\right]\leq E\left[\bar{V}^{\,\pi}_0/\bar{V}^{\,\tilde{\pi}}_0\right]=1$, for all $\pi\in\mathcal{A}$, thus showing that the set $\left\{\bar{V}^{\,\pi}_T/\bar{V}^{\,\tilde{\pi}}_T:\pi\in\mathcal{A}\right\}$ is bounded in $L^1$ and, hence, also in probability. Since the multiplication by the fixed random variable $\bar{V}^{\,\tilde{\pi}}_T$ does not affect the boundedness in probability, this implies that the NUPBR condition holds. Hence, due to Theorem \ref{ThmKK}, the financial market is viable.
The second assertion follows immediately from Proposition \ref{P2-2}.
\end{proof}

We emphasise again that all these results hold true even in the absence of an ELMM. For further comments on the relations between the GOP and the numeraire portfolio in a general semimartingale setting, we refer to Section 3 of \cite{KK} (see also \cite{HS} in the continuous semimartingale case).

\begin{Rem}[\emph{\textbf{On the GOP-denominated market}}]	\label{Rem-GO-num}
Due to Corollary \ref{GO-num}, the GOP coincides with the numeraire portfolio. Moreover, Corollary \ref{C2-1} shows that all portfolio processes $V^{\,\pi}$, for $\pi\in\mathcal{A}$, are local martingales when denominated in units of the GOP $V^{\,\pi^*}$. This means that, if we express all price processes in terms of the GOP, then the original probability measure $P$ becomes an ELMM for the GOP-denominated market. Hence, due to the fundamental theorem of asset pricing (see \cite{DS}), the classical (NFLVR) no-arbitrage-type condition holds for the GOP-denominated market. This observation suggests that the GOP-denominated market may be regarded as the minimal and natural setting for dealing with valuation and portfolio optimisation problems, even when there does not exist an ELMM for the original market $(S^0,S^1,\ldots,S^N)$ and this fact will be exploited in Section \ref{S4}. In a related context, see also \cite{CL}.
\end{Rem}

According to \cite{P0}, \cite{P1}, \cite{P2} and \cite{PH}, let us give the following Definition.

\begin{Def}	\label{D2-3}
For any portfolio process $V^{\,\pi}$, the process $\hat{V}^{\,\pi}=\bigl(\hat{V}^{\,\pi}_t\bigr)_{0\leq t \leq T}$, defined as $\hat{V}^{\,\pi}_t:=V_t^{\,\pi}/V_t^{\,\pi^*}$ for $t\in\left[0,T\right]$, is called \emph{benchmarked} portfolio process.	
A trading strategy $\pi\in\mathcal{A}$ and the associated portfolio process $V^{\,\pi}$ are said to be \emph{fair} if the benchmarked portfolio process $\hat{V}^{\,\pi}$ is a martingale.
We denote by $\mathcal{A}^F$ the set of all fair trading strategies in $\mathcal{A}$.
\end{Def}

According to Definition \ref{D2-3}, the result of Corollary \ref{C2-1} amounts to saying that all benchmarked portfolio processes are positive supermartingales. Note that every benchmarked portfolio process is a local martingale but not necessarily a true martingale. This amounts to saying that there may exist \emph{unfair} portfolios, namely portfolios for which the benchmarked value process is a strict local martingale.
The concept of benchmarking will become relevant in Section \ref{S4.2}, where we shall discuss its role for valuation purposes.

\begin{Rem}[\textbf{\emph{Other optimality properties of the GOP}}]
Besides maximising the growth-rate, the GOP enjoys several other optimality properties, many of which are illustrated in the monograph \cite{PH}. In particular, it has been shown that the GOP maximises the long-term growth rate among all admissible portfolios, see e.g. \cite{P2}. It is also well-known that the GOP is the solution to the problem of maximising an expected logarithmic utility function, see Section \ref{S4.3} and also \cite{KK}. Other interesting properties of the GOP include the impossibility of \emph{relative arbitrages} (or \emph{systematic outperformance}) with respect to it, see \cite{FK} and \cite{P2}, and, under suitable assumptions on the behavior of market participants, \emph{two-fund separation} results and connections with mean-variance efficiency, see e.g. \cite{P0} and \cite{P1}. Other properties of the growth-optimal strategy are also illustrated in the recent paper \cite{MTZ}.
\end{Rem}

\section{Replicating strategies and completeness of the financial market}	\label{S1.2}

Without relying on the existence of an ELMM for the financial market $\left(S^0,S^1,\ldots,S^N\right)$, in this Section we start laying the foundations for the valuation of arbitrary contingent claims. More specifically, in this Section we shall be concerned with the study of replicating (or hedging) strategies, formally defined as follows. 

\begin{Def}	\label{hedgstrat}
Let $H$ be a positive $\F$-measurable contingent claim (i.e. random variable) such that $E\left[\frac{\widehat{Z}_T}{S^0_T}H\right]<\infty$. If there exists a couple $\left(v^H,\pi^H\right)\in\left(0,\infty\right)\times\mathcal{A}$ such that $V^{\,v^H,\pi^H}_T=H$ $P$-a.s., then we say that $\pi^H$ is a replicating strategy for $H$.
\end{Def}

The following Proposition illustrates some basic features of a replicating strategy .

\begin{Prop}	\label{P1-2}
Suppose that Assumption \ref{A1-2} holds. Let $H$ be a positive $\F$-measurable contingent claim such that  $\E\left[\frac{\widehat{Z}_T}{S_T^0}\,H\right]<\infty$ and suppose there exists a trading strategy $\pi^H\in\mathcal{A}$ such that $V_T^{\,v^H,\pi^H}=H$ $P$-a.s. for $v^H=\E\left[\frac{\widehat{Z}_T}{S_T^0}\,H\right]$. Then the following hold:
\begin{itemize}
\item[(a)] the strategy $\pi^H$ is fair, in the sense of Definition \ref{D2-3};
\item[(b)] the strategy $\pi^H$ is unique, up to a null subset of $\Omega\times\left[0,T\right]$.
\end{itemize}  
Moreover, for every $\left(v,\pi\right)\in\left(0,\infty\right)\times\mathcal{A}$ such that $V_T^{\,v,\pi}=H$ $P$-a.s., we have $V^{\,v,\pi}_t\geq V^{\,v^H,\pi^H}_t$ $P$-a.s. for all $t\in\left[0,T\right]$. In particular, there cannot exist an element $\bar{\pi}\in\mathcal{A}$ such that $V_T^{\,\bar{v},\bar{\pi}}=H$ $P$-a.s. for some $\bar{v}<v^H$.
\end{Prop}
\begin{proof}
Corollary \ref{C2-1} implies that the benchmarked portfolio process $\hat{V}^{\,v^H,\pi^H}=\bigl(V^{\,v^H,\pi^H}_t/V^{\,\pi^*}_t\bigr)_{0\leq t \leq T}$ is a supermartingale. Moreover, it is also a martingale, due to the fact that: 
\be	\label{P2-1-1}
\hat{V}^{\,v^H,\pi^H}_0 = v^H
= E\biggl[\frac{\widehat{Z}_T}{S^0_T}H\biggr]
= E\left[\frac{V^{\,v^H,\pi^H}_T}{V^{\,\pi^*}_T}\right]
= E\left[\hat{V}^{\,v^H,\pi^H}_T\right]
\ee
where the third equality follows from Proposition \ref{P2-1}.
Part \emph{(a)} then follows from Definition \ref{D2-3}. To prove part \emph{(b)}, let $\hat{\pi}\in\mathcal{A}$ be a trading strategy such that $V_T^{\,v^H,\hat{\pi}}=H$ $P$-a.s. for $v^H=\E\left[\frac{\widehat{Z}_T}{S_T^0}\,H\right]$. 
Reasoning as in \eqref{P2-1-1}, the benchmarked portfolio process $\hat{V}^{\,v^H,\hat{\pi}}=\bigl(V^{\,v^H,\hat{\pi}}_t/V^{\,\pi^*}_t\bigr)_{0\leq t \leq T}$ is a martingale.
Together with the fact that $\hat{V}^{\,v^H,\hat{\pi}}_T=\frac{\widehat{Z}_T}{S^0_T}H=\hat{V}^{\,v^H,\pi^H}_T$ $P$-a.s., this implies that $V^{\,v^H,\hat{\pi}}_t=V^{\,v^H,\pi^H}_t$ $P$-a.s. for all $t\in\left[0,T\right]$. Part \emph{(b)} then follows by the same arguments as in the second part of the proof of Proposition \ref{P2-2}.
To prove the last assertion let $\left(v,\pi\right)\in\left(0,\infty\right)\times\mathcal{A}$ be such that $V_T^{\,v,\pi}=H$ $P$-a.s. Due to Corollary \ref{C2-1}, the benchmarked portfolio process $\hat{V}^{\,v,\pi}=\bigl(V^{\,v,\pi}_t/V^{\,\pi^*}_t\bigr)_{0\leq t \leq T}$ is a supermartingale. So, for any $t\in\left[0,T\right]$, due to part \emph{(a)}:
$$
\hat{V}^{\,v^H,\pi^H}_t = E\left[\hat{V}^{\,v^H,\pi^H}_T\bigr|\F_t\right]
= E\biggl[\frac{\widehat{Z}_T}{S^0_T}H\Bigr|\F_t\biggr]
= E\left[\hat{V}^{\,v,\pi}_T\bigr|\F_t\right]
\leq \hat{V}^{\,v,\pi}_t
\qquad \text{$P$-a.s.}
$$
and, hence, $V^{\,v^H,\pi^H}_t\leq V^{\,v,\pi}_t$ $P$-a.s. for all $t\in\left[0,T\right]$. For $t=0$, this implies that $v\geq v^H$, thus completing the proof.
\end{proof} 

\begin{Rem}
Observe that Proposition \ref{P1-2} does not exclude the existence of a trading strategy $\check{\pi}\in\mathcal{A}$ such that $V_T^{\,\check{v},\check{\pi}}=H$ $P$-a.s. for some $\check{v}>v^H$. However, one can argue that it may not be optimal to invest in such a strategy in order to replicate $H$, since it requires a larger initial investment and leads to an unfair portfolio process. Indeed, Proposition \ref{P1-2} shows that $v^H=E\left[\frac{\widehat{Z}_T}{S_T^0}\,H\right]$ is the minimal initial capital starting from which one can replicate the contingent claim $H$. To this effect, see also Remark 1.6.4 in \cite{KS2}.
\end{Rem}

A particularly nice and interesting situation arises when the financial market is \emph{complete}, meaning that every \emph{contingent claim} can be perfectly replicated starting from some initial investment by investing in the financial market according to some admissible self-financing trading strategy. 

\begin{Def}	\label{D1-3}
The financial market $\left(S^0,S^1,\ldots,S^N\right)$ is said to be \emph{complete} if for any positive $\F$-measurable contingent claim $H$ such that $\E\left[\frac{\widehat{Z}_T}{S_T^0}\,H\right]<\infty$ there exists a couple $\left(v^H,\pi^H\right)\in\left(0,\infty\right)\times\mathcal{A}$ such that $V_T^{\,v^H,\pi^H}=H$ $P$-a.s.
\end{Def}

In general, the financial market described in Section \ref{S1.1} is incomplete and, hence, not all contingent claims can be perfectly replicated. The following Theorem gives a sufficient condition for the financial market to be complete. The proof is similar to that of Theorem 1.6.6 in \cite{KS2}, except that we avoid the use of any ELMM, since the latter may fail to exist in our general context. This allows us to highlight the fact that the concept of market completeness does not depend on the existence of an ELMM. 

\begin{Thm}	\label{T1-1}
Suppose that Assumption \ref{A1-2} holds. Assume furthermore that $\FF=\FF^W$, where $\FF^W$ is the $P$-augmented Brownian filtration associated to $W$, and that $d=N$. Then the financial market $\left(S^0,S^1,\ldots,S^N\right)$ is complete. More precisely, any positive $\F$-measurable contingent claim $H$ with $E\left[\frac{\widehat{Z}_T}{S^0_T}H\right]<\infty$ can be replicated by a fair portfolio process $V^{\,v^H,\pi^H}$, with $v^H=E\left[\frac{\widehat{Z}_T}{S^0_T}H\right]$ and $\pi^H\in\mathcal{A}^F$.
\end{Thm}
\begin{proof}
Let $H$ be a positive $\F=\F_T^W$-measurable random variable such that $\E\left[\frac{\widehat{Z}_T}{S_T^0}\,H\right]<\infty$ and define the martingale $M=\left(M_t\right)_{0\leq t \leq T}$ by $M_t:=\E\left[\frac{\widehat{Z}_T}{S_T^0}\,H\bigl|\F_t\right]$, for $t\in\left[0,T\right]$. According to the martingale representation theorem (see \cite{KS1}, Theorem 3.4.15 and Problem 3.4.16) there exists an $\R^N$-valued progressively measurable process $\varphi=\left(\varphi_t\right)_{0\leq t \leq T}$ such that $\int_0^T\left\|\varphi_t\right\|^2dt<\infty$ $P$-a.s. and:
\be	\label{T1-1-1}
M_t = M_0+\int_0^t\varphi_u'\,dW_u	\qquad \text{for all } t\in\left[0,T\right]
\ee
Define then the positive process $V=\left(V_t\right)_{0\leq t \leq T}$ by $V_t:=\frac{S_t^0}{\widehat{Z}_t}M_t$, for $t\in\left[0,T\right]$. Recalling that $S_0^0=1$, we have $v^H:=V_0=M_0=\E\left[\frac{\widehat{Z}_T}{S_T^0}\,H\right]$. 
The standing Assumption \ref{A1-1}, together with the fact that $d=N$, implies that the matrix $\sigma_t$ is $P$-a.s. invertible for all $t\in\left[0,T\right]$. 
Then, an application of the product rule together with equations \eqref{E1-9} and \eqref{T1-1-1}, gives:
\be	\label{T1-1-3}	\ba
d\left(\frac{V_t}{S^0_t}\right) = d\left(\frac{M_t}{\widehat{Z}_t}\right)
&=M_t\,d\frac{1}{\widehat{Z}_t}+\frac{1}{\widehat{Z}_t}\,dM_t+d\Bigl\langle M,\frac{1}{\widehat{Z}}\Bigr\rangle_t	\\
&= \frac{M_t}{\widehat{Z}_t}\,\theta_t'\,dW_t+\frac{M_t}{\widehat{Z}_t}\,\left\|\theta_t\right\|^2dt
+\frac{1}{\widehat{Z}_t}\,\varphi_t'\,dW_t+\frac{1}{\widehat{Z}_t}\,\varphi_t'\,\theta_t\,dt	\\
&= \frac{V_t}{S_t^0}\left(\theta_t+\frac{\varphi_t}{M_t}\right)'\theta_t\,dt
+\frac{V_t}{S_t^0}\left(\theta_t+\frac{\varphi_t}{M_t}\right)'dW_t	\\
&= \frac{V_t}{S_t^0}\left(\theta_t+\frac{\varphi_t}{M_t}\right)'\sigma^{-1}_t\left(\mu_t-r_t\mathbf{1}\right)dt
+\frac{V_t}{S_t^0}\left(\theta_t+\frac{\varphi_t}{M_t}\right)'\sigma^{-1}_t\,\sigma_t\,dW_t	\\
&= \frac{V_t}{S_t^0}\sum_{i=1}^N\,\pi_t^{H,i}\,\frac{d\bar{S}_t^i}{\bar{S}_t^i}
\ea	\ee
where $\pi_t^H=\bigl(\pi_t^{H,1},\ldots,\pi_t^{H,N}\bigr)':=\left(\sigma_t'\right)^{-1}\bigl(\theta_t+\frac{\varphi_t}{M_t}\bigr)$, for all $t\in\left[0,T\right]$. 
The last line of \eqref{T1-1-3} shows that the process $\bar{V}:=V/S^0=\left(V_t/S^0_t\right)_{0\leq t \leq T}$ can be represented as a stochastic exponential as in part \emph{(b)} of Definition \ref{D1-1}. Hence, it remains to check that the process $\pi^H$ satisfies the integrability conditions of part \emph{(a)} of Definition \ref{D1-1}. Due to Lemma \ref{L1-1}, it suffices to verify that $\int_0^T\left\|\sigma_t'\,\pi_t^H\right\|^2dt<\infty$ $P$-a.s. This can be shown as follows:
$$
\int_0^T\left\|\sigma'_t\,\pi_t^H\right\|^2dt
=\int_0^T\left\|\theta_t+\frac{\varphi_t}{M_t}\right\|^2dt	
\leq 2\int_0^T\left\|\theta_t\right\|^2dt+2\,\left\|\frac{1}{M}\right\|_{\infty}\int_0^T\left\|\varphi_t\right\|^2dt
<\infty	\qquad \text{$P$-a.s.}
$$
due to Assumption \ref{A1-2} and because $\left\|\frac{1}{M}\right\|_{\infty}:=\underset{t\in\left[0,T\right]}{\max}\left|\frac{1}{M_t}\right|<\infty$ $P$-a.s. due to the continuity of $M$.
We have thus shown that $\pi^H$ is an admissible trading strategy, i.e. $\pi^H\in\mathcal{A}$, and the associated portfolio process $V^{\,v^H,\pi^H}=\bigl(V^{\,v^H,\pi^H}_t\bigr)_{0\leq t \leq T}$ satisfies $V_T^{\,v^H,\pi^H}=V_T=H$ $P$-a.s. with $v^H=\E\left[\frac{\widehat{Z}_T}{S_T^0}\,H\right]$. Furthermore, since $\hat{V}^{\,v^H,\pi^H}_t=V^{\,v^H,\pi^H}_t/V^{\,\pi^*}_t=V_t\,\widehat{Z}_t/S_t^0=M_t$, we also have $\pi^H\in\mathcal{A}^F$.
\end{proof}

We close this Section with some important comments on the result of Theorem \ref{T1-1}.

\begin{Rem} \label{Rem-S1.2} $ $
\begin{enumerate}
\item
We want to emphasise that Theorem \ref{T1-1} does not rely on the existence of an ELMM for the financial market $\left(S^0,S^1,\ldots,S^N\right)$. This amounts to saying that the completeness of a financial market does not necessarily imply that some mild forms of arbitrage opportunities are \emph{a priori} excluded. 
Typical ``textbook versions'' of the so-called \emph{second Fundamental Theorem of Asset Pricing} state that the completeness of the financial market is equivalent to the uniqueness of the \emph{Equivalent (Local) Martingale Measure}, loosely speaking. However, Theorem \ref{T1-1} shows that we can have a complete financial market even when no E(L)MM exists at all. The fact that absence of arbitrage opportunities and market completeness should be regarded as distinct concepts has been already pointed out in a very general setting in \cite{JM}.
The completeness of the financial market model will play a crucial role in Section \ref{S4}, where we shall be concerned with valuation and hedging problems in the absence of an ELMM.
\item
Following the reasoning in the proof of Theorem 1.6.6 of \cite{KS2}, but avoiding the use of an ELMM (which in our context may fail to exist), it is possible to prove a converse result to Theorem \ref{T1-1}. More precisely, if we assume that $\FF=\FF^W$ and that every $\F$-measurable positive random variable $H$ with $v^H:=E\left[\frac{\widehat{Z}_T}{S^0_T}H\right]<\infty$ admits a trading strategy $\pi^H\in\mathcal{A}$ such that $V^{\,v^H,\pi^H}_T=H$ $P$-a.s., then we necessarily have $d=N$. Moreover, it can be shown that the completeness of the financial market is \emph{equivalent} to the existence of a unique martingale deflator and this holds true even in more general models based on continuous semimartingales. For details, we refer the interested reader to Chapter 4 of \cite{Fo}.
\end{enumerate}
\end{Rem}

\section{Contingent claim valuation without ELMMs}	\label{S4}

The main goal of this Section is to show how one can proceed to the valuation of contingent claims in financial market models which may not necessarily admit an ELMM. Since the non-existence of a properly defined martingale measure precludes the whole machinery of risk-neutral pricing, this appears as a non-trivial issue. Here we concentrate on the situation of a complete financial market, as considered at the end of the last Section (see Section \ref{Concl} for possible extensions to incomplete markets). A major focus of this Section is on providing a mathematical justification for the so-called \emph{real-world pricing approach}, according to which the valuation of contingent claims is performed under the original (or \emph{real-world}) probability measure $P$ using the GOP as the natural numeraire. 

\begin{Rem}
In this Section we shall be concerned with the problem of \emph{pricing} contingent claims. However, one should be rather careful with the terminology and distinguish between a \emph{value} assigned to a contingent claim and its prevailing \emph{market price}. Indeed, the former represents the outcome of an a priori chosen valuation rule, while the latter is the price determined by supply and demand forces in the financial market. Since the choice of the valuation criterion is a subjective one, the two concepts of \emph{value} and \emph{market price} do not necessarily coincide. This is especially true when arbitrage opportunities and/or bubble phenomena are not excluded from the financial market. In this Section, we use the word ``price'' only to be consistent with the standard terminology in the literature.
\end{Rem}

\subsection{Real-world pricing and the \emph{benchmark approach}}	\label{S4.2}

We start by introducing the concept of \emph{real-world price}, which is at the core of the so-called \emph{benchmark approach} to the valuation of contingent claims.

\begin{Def}	\label{D4-2}
Let $H$ be a positive $\F$-measurable contingent claim such that $\E\left[\frac{\widehat{Z}_T}{S_T^0}\,H\right]<\infty$. 
If there exists a fair portfolio process $V^{\,v^H,\pi^H}=\bigl(V^{\,v^H,\pi^H}_t\bigr)_{0\leq t \leq T}$ such that $V_T^{\,v^H,\pi^H}=H$ $P$-a.s., for some $\left(v^H,\pi^H\right)\in\left(0,\infty\right)\times\mathcal{A}^F$, then the \emph{real-world price} of $H$ at time $t$, denoted as $\Pi_t^H$, is defined as follows:
\be	\label{E4-1}
\Pi_t^H := V_t^{\,\pi^*}\,\E\left[\frac{H}{V_T^{\,\pi^*}}\Bigl|\F_t\right]
\ee
for every $t\in\left[0,T\right]$ and where $V^{\,\pi^*}=\left(V_t^{\,\pi^*}\right)_{0\leq t \leq T}$ denotes the GOP.
\end{Def}

The terminology \emph{real-world price} is used to indicate that, unlike in the traditional setting, all contingent claims are valued under the original real-world probability measure $P$ and not under an equivalent risk-neutral measure. This allows us to extend the valuation of contingent claims to financial markets for which no ELMM may exist. The concept of \emph{real-world price} gives rise to the so-called \emph{benchmark approach} to the valuation of contingent claims in view of the fact that the GOP plays the role of the natural numeraire portfolio (compare Remark \ref{Rem-GO-num}). For this reason we shall refer to it as the \emph{benchmark} portfolio. We refer the reader to \cite{P1}, \cite{P2} and \cite{PH} for a thorough presentation of the benchmark approach.

Clearly, if there exists a fair portfolio process $V^{\,v^H,\pi^H}$ such that $V^{\,v^H,\pi^H}_T=H$ $P$-a.s. for $\left(v^H,\pi^H\right)\in\left(0,\infty\right)\times\mathcal{A}^F$, then the real-world price coincides with the value of the fair replicating portfolio. In fact, for all $t\in\left[0,T\right]$:
$$
\Pi^H_t = V^{\,\pi^*}_tE\left[\frac{H}{V^{\,\pi^*}_T}\Bigr|\F_t\right]
= V^{\,\pi^*}_t E\biggl[\frac{V^{\,v^H,\pi^H}_T}{V^{\,\pi^*}_T}\Bigr|\F_t\biggr]
= V^{\,v^H,\pi^H}_t
\qquad \text{$P$-a.s.}
$$
where the last equality is due to the fairness of $V^{\,v^H,\pi^H}$, see Definition \ref{D2-3}. Moreover, the second part of Proposition \ref{P1-2} gives an economic rationale for the use of the real-world pricing formula \eqref{E4-1}, since it shows that the latter gives the value of the least expensive replicating portfolio. This property has been called the \emph{law of the minimal price} (see \cite{P2}, Section 4). The following simple Proposition immediately follows from Theorem \ref{T1-1}.

\begin{Prop}	\label{P4-2}
Suppose that Assumption \ref{A1-2} holds.
Let $H$ be a positive $\F$-measurable contingent claim such that $\E\left[\frac{\widehat{Z}_T}{S^0_T}\,H\right]<\infty$. 
Then, under the assumptions of Theorem \ref{T1-1}, the following hold:
\begin{itemize}
\item[(a)] there exists a fair portfolio process $V^{\,v^H,\pi^H}=\bigl(V_t^{\,v^H,\pi^H}\bigr)_{0\leq t \leq T}$ such that $V^{\,v^H,\pi^H}_T=H$ $P$-a.s.;
\item[(b)] the real-world price (at time $t=0$) is given by $\Pi^H_0=E\left[\frac{H}{V^{\,\pi^*}_T}\right]=E\left[\frac{\widehat{Z}_T}{S^0_T}H\right]=v^H$.
\end{itemize}
\end{Prop}

\clearpage

\begin{Rem} $ $
\begin{enumerate}
\item
Notice that, due to Proposition \ref{P2-1}, the real-world pricing formula \eqref{E4-1} can be rewritten as follows, for any $t\in\left[0,T\right]$:
\be	\label{E4-2}
\Pi_t^H=\frac{S_t^0}{\widehat{Z}_t}\,\E\biggl[\frac{\widehat{Z}_T}{S_T^0}\,H\Bigr|\F_t\biggr]
\ee
Suppose now that $\E\bigl[\widehat{Z}_T\bigr]=1$, so that $\widehat{Z}$ represents the density process of the ELMM $\widehat{Q}$ (see Remark \ref{Rem-S1.1bis}). Due to the Bayes formula, equation \eqref{E4-2} can then be rewritten as follows:
$$
\Pi_t^H=S_t^0\,\E^{\widehat{Q}}\left[\frac{H}{S_T^0}\Bigl|\F_t\right]
$$
and we recover the usual risk-neutral pricing formula (see also \cite{P2}, Section 5, and \cite{PH}, Section 10.4). In this sense, the real-world pricing approach can be regarded as a consistent extension of the usual risk-neutral valuation approach to a financial market for which an ELMM may fail to exist.
\item
Let us suppose for a moment that $H$ and the final value of the GOP $V_T^{\,\pi^*}$ are conditionally independent given the $\sigma$-field $\F_t$, for all $t\in\left[0,T\right]$. The real-world pricing formula \eqref{E4-1} can then be rewritten as follows:
\be	\label{E4-3}
\Pi_t^H=V_t^{\,\pi^*}\,\E\left[\frac{1}{V_T^{\,\pi^*}}\Bigl|\F_t\right]\E\left[H|\F_t\right] 
=: P\left(t,T\right)\E\left[H|\F_t\right]
\ee
where $P\left(t,T\right)$ denotes the \emph{fair value} at time $t$ of a \emph{zero coupon bond} with maturity $T$ (i.e. a contingent claim which pays the deterministic amount $1$ at time $T$). 
This shows that, under the (rather strong) assumption of conditional independence, one can recover the well-known \emph{actuarial pricing formula} (see also \cite{P1}, Corollary 3.4, and \cite{P2}, Section 5).
\item
We want to point out that part \emph{(b)} of Proposition \ref{P4-2} can be easily generalised to any time $t\in\left[0,T\right]$; compare for instance Proposition 10 in \cite{RG}.
\end{enumerate}
\end{Rem}

In view of the above Remarks, it is interesting to observe how several different valuation approaches which have been widely used in finance and insurance, such as risk-neutral pricing and actuarial pricing, are both generalised and unified under the concept of real-world pricing. We refer to Section 10.4 of \cite{PH} for related comments on the unifying aspects of the benchmark approach.

\subsection{The \emph{upper hedging price} approach}	\label{S4.1}

The \emph{upper hedging price} (or \emph{super-hedging price}) is a classical approach to the valuation of contingent claims (see e.g. \cite{KS2}, Section 5.5.3). The intuitive idea is to find the smallest initial capital which allows one to obtain a final wealth which is greater or equal than the payoff at maturity of a given contingent claim. 

\begin{Def}	\label{D4-1} 
Let $H$ be a positive $\F$-measurable contingent claim. The \emph{upper hedging price} $\mathcal{U}\left(H\right)$ of $H$ is defined as follows:
$$
\mathcal{U}\left(H\right):=\inf\bigl\{v\in\left[0,\infty\right):\exists\;\pi\in\mathcal{A}\text{ such that } V_T^{\,v,\pi}\geq H \text{ $P$-a.s.}\bigr\}
$$
with the usual convention $\inf\emptyset=\infty$.
\end{Def}

The next Proposition shows that, in a complete diffusion-based financial market, the upper hedging price takes a particularly simple and natural form. This result is an immediate consequence of the supermartingale property of benchmarked portfolio processes together with the completeness of the financial market but, for the reader's convenience, we give a detailed proof.

\begin{Prop}	\label{T4-1}
Let $H$ be a positive $\F$-measurable contingent claim such that $\E\left[\frac{\widehat{Z}_T}{S_T^0}\,H\right]<\infty$. Then, under the assumptions of Theorem \ref{T1-1}, the upper hedging price of $H$ is explicitly given by:
\be	\label{T4-1-1}
\mathcal{U}\left(H\right)=\E\biggl[\frac{\widehat{Z}_T}{S_T^0}\,H\biggr]
\ee
\end{Prop}
\begin{proof}
In order to prove \eqref{T4-1-1}, we show both directions of inequality.
\begin{itemize}
\item[($\,\geq$):]
If $\left\{v\in\left[0,\infty\right):\exists\;\pi\in\mathcal{A}\text{ such that } V_T^{\,v,\pi}\geq H \text{ $P$-a.s.}\right\}=\emptyset$ then we have $\E\left[\frac{\widehat{Z}_T}{S_T^0}\,H\right]<\mathcal{U}\left(H\right)=\infty$. So, let us assume there exists a couple $\left(v,\pi\right)\in\left[0,\infty\right)\times\mathcal{A}$ such that $V_T^{\,v,\pi}\geq H$ $P$-a.s. Under Assumption \ref{A1-2}, due to Corollary \ref{C2-1}, the benchmarked portfolio process $\hat{V}^{\,v,\pi}=\bigl(V^{\,v,\pi}_t/V^{\,\pi^*}_t\bigr)_{0\leq t \leq T}$ is a supermartingale and so, recalling also Proposition \ref{P2-1}:
$$
v = \hat{V}_0^{\,v,\pi} \geq \E\bigl[\hat{V}_T^{\,v,\pi}\bigr]
=\E\biggl[\frac{\widehat{Z}_T}{S_T^0}\,V_T^{\,v,\pi}\biggr]
\geq \E\biggl[\frac{\widehat{Z}_T}{S_T^0}\,H\biggr]
$$
This implies that $\mathcal{U}\left(H\right)\geq\E\left[\frac{\widehat{Z}_T}{S_T^0}\,H\right]$.
\item[($\,\leq$):]
Under the present assumptions, Theorem \ref{T1-1} yields the existence of a couple $\left(v^H,\pi^H\right)\in\left(0,\infty\right)\times\mathcal{A}^F$ such that $V_T^{\,v^H,\pi^H}=H$ $P$-a.s. and where $v^H=\E\left[\frac{\widehat{Z}_T}{S_T^0}\,H\right]$. 
Hence:
$$
\E\biggl[\frac{\widehat{Z}_T}{S_T^0}\,H\biggr]=v^H\in
\bigl\{v\in\left[0,\infty\right):\exists\;\pi\in\mathcal{A}\text{ such that } V_T^{\,v,\pi}\geq H \text{ $P$-a.s.}\bigr\}
$$ 
This implies that $\mathcal{U}\left(H\right)\leq\E\left[\frac{\widehat{Z}_T}{S_T^0}\,H\right]$.
\end{itemize}
\end{proof}

An analogous result can be found in Proposition 5.3.2 of \cite{KS2} (compare also \cite{FK}, Section 10). We want to point out that Definition \ref{D4-1} can be easily generalised to an arbitrary time point $t\in\left[0,T\right]$ in order to define the upper hedging price at $t\in\left[0,T\right]$. The result of Proposition \ref{T4-1} carries over to this slightly generalised setting with essentially the same proof, see e.g. Theorem 3 in \cite{RG}.

\begin{Rem} $ $
\begin{enumerate}
\item
Notice that, due to Proposition \ref{P2-1}, equation \eqref{T4-1-1} can be rewritten as follows:
$$
\mathcal{U}\left(H\right)=\E\biggl[\frac{\widehat{Z}_T}{S_T^0}H\biggr]
=\E\biggl[\frac{H}{V_T^{\,\pi^*}}\biggr]
$$
This shows that the upper hedging price can be obtained by computing the expectation of the benchmarked value (in the sense of Definition \ref{D2-3}) of the contingent claim $H$ under the real-world probability measure $P$ and thus coincides with the real-world price (evaluated at $t=0$), see part \emph{(b)} of Proposition \ref{P4-2}. 
\item
Suppose that $\E\bigl[\widehat{Z}_T\bigr]=1$. As explained in Remark \ref{Rem-S1.1bis}, the process $\widehat{Z}$ represents then the density process of the ELMM $\widehat{Q}$. In this case, the upper hedging price $\mathcal{U}\left(H\right)$ yields the usual risk-neutral valuation formula, i.e. we have $\mathcal{U}\left(H\right)=\E^{\widehat{Q}}\left[H/S_T^0\right]$. 
\end{enumerate}
\end{Rem}

\subsection{Utility indifference valuation}	\label{S4.3}

The real-world valuation approach has been justified so far on the basis of replication arguments, as can be seen from Propositions \ref{P4-2} and \ref{T4-1}. We now present a different approach which uses the idea of \emph{utility indifference valuation}. To this effect, let us first consider the problem of maximising an expected utility function of the discounted final wealth. Recall that, due to Theorem \ref{ThmKK}, we can meaningfully consider portfolio optimisation problems even in the absence of an ELMM for $\left(S^0,S^1,\ldots,S^N\right)$.

\begin{Def}	\label{D4-3}
We call \emph{utility function} $U$ a function $U:\left[0,\infty\right)\rightarrow\left[0,\infty\right)$ such that:
\begin{enumerate}
\item $U$ is strictly increasing and strictly concave, continuously differentiable;
\item $\underset{x\rightarrow\infty}{\lim}U'\left(x\right)=0$ and $\underset{x\rightarrow 0}{\lim}\,U'\left(x\right)=\infty$.
\end{enumerate}
\end{Def}

\begin{Prob}[\textbf{expected utility maximisation}]
Let $U$ be as in Definition \ref{D4-3} and let $v\in\left(0,\infty\right)$. The expected utility maximisation problem consists in the following:
\be	\label{ProbEUM}
\text{maximise $\E\left[U\left(\bar{V}_T^{\,v,\pi}\right)\right]$ over all $\pi\in\mathcal{A}$}
\ee
\end{Prob}

The following Lemma shows that, in the case of a complete financial market, there is no loss of generality in restricting our attention to fair strategies only. Recall that, due to Definition \ref{D2-3}, $\mathcal{A}^F$ denotes the set of all fair trading strategies in $\mathcal{A}$.

\begin{Lem}	\label{fair}
Under the assumptions of Theorem \ref{T1-1}, for any utility function $U$ and for any $v\in\left(0,\infty\right)$, the following holds:
\be	\label{faireq}
\underset{\pi\in\mathcal{A}}{\sup}\,E\left[U\left(\bar{V}^{\,v,\pi}_T\right)\right]
= \underset{\pi\in\mathcal{A}^F}{\sup}E\left[U\left(\bar{V}^{\,v,\pi}_T\right)\right]
\ee
\end{Lem}
\begin{proof}
It is clear that ``$\geq$'' holds in \eqref{faireq}, since $\mathcal{A}^F\subseteq\mathcal{A}$. To show the reverse inequality, let us consider an arbitrary strategy $\pi\in\mathcal{A}$. The benchmarked portfolio process $\hat{V}^{\,v,\pi}=\bigl(V^{\,v,\pi}_t/V^{\,\pi^*}_t\bigr)_{0\leq t \leq T}$ is a supermartingale, due to Corollary \ref{C2-1}, and hence:
$$
v':=E\biggl[\frac{\widehat{Z}_T}{S^0_T}V^{\,v,\pi}_T\biggr]
= E\left[\frac{V^{\,v,\pi}_T}{V^{\,\pi^*}_T}\right] 
\leq v
$$
with equality holding if and only if $\pi\in\mathcal{A}^F$. Let $\bar{v}:=v-v'\geq 0$. It is clear that the positive $\F$-measurable random variable $\bar{H}:=\bar{V}^{\,v,\pi}_T+\bar{v}/\widehat{Z}_T$ satisfies $E\bigl[\widehat{Z}_T\bar{H}\bigr]=v$ and so, due to Theorem \ref{T1-1}, there exists an element $\pi^H\in\mathcal{A}^F$ such that $\bar{V}^{\,v,\pi^H}_T=\bar{H}\geq\bar{V}^{\,v,\pi}_T$ $P$-a.s., with equality holding if and only if the strategy $\pi$ is fair. We then have, due to the monotonicity of $U$:
$$
E\bigl[U\left(\bar{V}^{\,v,\pi}_T\right)\bigr]
\leq E\bigl[U\left(\bar{H}\right)\bigr]
= E\bigl[U\bigl(\bar{V}^{\,v,\pi^H}_T\bigr)\bigr]
\leq \underset{\pi\in\mathcal{A}^F}{\sup}E\bigl[U\left(\bar{V}^{\,v,\pi}_T\right)\bigr]
$$
Since $\pi\in\mathcal{A}$ was arbitrary, this shows the ``$\leq$'' inequality in \eqref{faireq}.
\end{proof}

In particular, Lemma \ref{fair} shows that, in the context of portfolio optimisation problems, restricting the class of admissible trading strategies to \emph{fair} admissible strategies is not only ``reasonable'', as argued in Chapter 11 of \cite{PH}, but exactly yields the same optimal value of the problem in its original formulation. 
The following Theorem gives the solution to Problem \eqref{ProbEUM}, in the case of a complete financial market. Related results can be found in Lemma 5 of \cite{RG} and Theorem 3.7.6 of \cite{KS2}.

\begin{Thm}	\label{T4-2}
Let the assumptions of Theorem \ref{T1-1} hold and let $U$ be a utility function. For $v\in\left(0,\infty\right)$, assume that the function $\mathcal{W}\left(y\right):=\E\left[\widehat{Z}_T\,I\bigl(y/\bar{V}_T^{\,v,\pi^*}\bigr)\right]$ is finite for every $y\in\left(0,\infty\right)$, where $I$ is the inverse function of $U'$.
Then the function $\mathcal{W}$ is invertible and the optimal discounted final wealth $\bar{V}^{\,v,\pi^U}_T$ for Problem \eqref{ProbEUM} is explicitly given as follows:
\be	\label{T4-2-1}
\bar{V}^{\,v,\pi^U}_T = I\left(\frac{\mathcal{Y}\left(v\right)}{\bar{V}_T^{\,v,\pi^*}}\right)
\ee
where $\mathcal{Y}$ denotes the inverse function of $\mathcal{W}$. The optimal strategy $\pi^U\in\mathcal{A}^F$ is given by the replicating strategy for the right hand side of \eqref{T4-2-1}.
\end{Thm}
\begin{proof}
Note first that, due to Definition \ref{D4-3}, the function $U'$ admits a strictly decreasing continuous inverse function $I:\left[0,\infty\right]\rightarrow\left[0,\infty\right]$ with $I\left(0\right)=\infty$ and $I\left(\infty\right)=0$. We have then the following well-known result from convex analysis (see e.g. \cite{KS2}, Section 3.4):
\be	\label{T4-2-2}
U\bigl(I\left(y\right)\bigr)-yI\left(y\right) \geq U\left(x\right)-xy
\qquad \text{ for }\quad 0\leq x<\infty,\; 0<y<\infty
\ee
As in Lemma 3.6.2 of \cite{KS2}, it can be shown that the function $\mathcal{W}:\left[0,\infty\right]\rightarrow\left[0,\infty\right]$ is strictly decreasing and continuous and, hence, it admits an inverse function $\mathcal{Y}:\left[0,\infty\right]\rightarrow\left[0,\infty\right]$. Since $\mathcal{W}\bigl(\mathcal{Y}\left(v\right)\bigr)=v$, for any $v\in\left(0,\infty\right)$, Theorem \ref{T1-1} shows that there exists a fair strategy $\pi^U\in\mathcal{A}^F$ such that $\bar{V}^{\,v,\pi^U}_T=I\bigl(\mathcal{Y}\left(v\right)/\bar{V}_T^{\,v,\pi^*}\bigr)$ $P$-a.s. Furthermore, for any $\pi\in\mathcal{A}^F$, the inequality \eqref{T4-2-2} with $y=\mathcal{Y}\left(v\right)/\bar{V}_T^{\,v,\pi^*}$ and $x=\bar{V}^{\,v,\pi}_T$ gives that:
$$	\ba
E\left[U\left(\bar{V}^{\,v,\pi^U}_T\right)\right]
&= E\left[U\left(I\left(\frac{\mathcal{Y}\left(v\right)}{\bar{V}_T^{\,v,\pi^*}}\right)\right)\right]
\geq E\bigl[U\left(\bar{V}^{\,v,\pi}_T\right)\bigr]
+\mathcal{Y}\left(v\right)
E\left[\frac{1}{\bar{V}^{\,v,\pi^*}_T}\left(I\left(\frac{\mathcal{Y}\left(v\right)}{\bar{V}_T^{\,v,\pi^*}}\right)-\bar{V}^{\,v,\pi}_T\right)\right] \\
&= E\bigl[U\left(\bar{V}^{\,v,\pi}_T\right)\bigr]
+\mathcal{Y}\left(v\right)
E\left[\frac{1}{\bar{V}^{\,v,\pi^*}_T}\left(\bar{V}^{\,v,\pi^U}_T-\bar{V}^{\,v,\pi}_T\right)\right]
= E\bigl[U\left(\bar{V}^{\,v,\pi}_T\right)\bigr]
\ea	$$
thus showing that, based also on Lemma \ref{fair}, $\pi^U\in\mathcal{A}^F$ solves Problem \eqref{ProbEUM}.
\end{proof}

\mbox{}

\begin{Rem} $ $
\begin{enumerate}
\item
It is important to observe that Theorem \ref{T4-2} does not rely on the existence of an ELMM. This amounts to saying that we can meaningfully solve expected utility maximisation problems even when no ELMM exists or, equivalently, when the traditional (NFLVR) no-arbitrage-type condition fails to hold. The crucial assumption for the validity of Theorem \ref{T4-2} is Assumption \ref{A1-2}, which ensures that the financial market is viable, in the sense that there are no arbitrages of the first kind (compare Theorem \ref{ThmKK} and Corollary \ref{NOarb}). 
\item
The assumption that the function $\mathcal{W}\left(y\right):=\E\left[\widehat{Z}_T\,I\bigl(y/\bar{V}_T^{\,v,\pi^*}\bigr)\right]$ be finite for every $y\in\left(0,\infty\right)$ can be replaced by suitable technical conditions on the utility function $U$ and on the processes $\mu$ and $\sigma$ (see Remarks 3.6.8 and 3.6.9 in \cite{KS2} for more details).
\end{enumerate}
\end{Rem}

Having solved in general the expected utility maximisation problem, we are now in a position to give the definition of the \emph{utility indifference price}, in the spirit of \cite{D} (compare also \cite{RG}, Section 4.2, \cite{PH}, Definition 11.4.1, and \cite{PR2}, Definition 10)\footnote{In \cite{RG} and \cite{PR2} the authors generalise Definition \ref{D4-4} to an arbitrary time $t\in\left[0,T\right]$. However, since the results and the techniques remain essentially unchanged, we only consider the basic case $t=0$.}. Until the end of this Section, we let $U$ be a utility function, in the sense of Definition \ref{D4-3}, such that all expected values below exist and are finite.

\begin{Def}	\label{D4-4}
Let $H$ be a positive $\F$-measurable contingent claim and $v\in\left(0,\infty\right)$. For $p\geq 0$, let us define, for a given utility function $U$, the function $W^U_p:\left[0,1\right]\rightarrow\left[0,\infty\right)$ as follows:
\be	\label{E4-4}
W_p^U\left(\varepsilon\right) := \E\left[U\left(\left(v-\varepsilon p\right)\bar{V}_T^{\,\pi^U}+\,\varepsilon\bar{H}\right)\right]
\ee
where $\pi^U\in\mathcal{A}^F$ solves Problem \eqref{ProbEUM} for the utility function $U$. 
The \emph{utility indifference price} of the contingent claim $H$ is defined as the value $p\left(H\right)$ which satisfies the following condition:
\be	\label{E4-5}
\underset{\varepsilon\rightarrow 0}{\lim}\,
\frac{W_{p\left(H\right)}^U\left(\varepsilon\right)-W_{p\left(H\right)}^U\left(0\right)}{\varepsilon}=0
\ee
\end{Def}

Definition \ref{D4-4} is based on a ``marginal rate of substitution'' argument, as first pointed out in \cite{D}. In fact, $p\left(H\right)$ can be thought of as the value at which an investor is marginally indifferent between the two following alternatives:
\begin{itemize}
\item invest an infinitesimal part $\varepsilon p\left(H\right)$ of the initial endowment $v$ into the contingent claim $H$ and invest the remaining wealth $\bigl(v-\varepsilon p\left(H\right)\bigr)$ according to the optimal trading strategy $\pi^U$;
\item ignore the contingent claim $H$ and simply invest the whole initial endowment $v$ according to the optimal trading strategy $\pi^U$.
\end{itemize}  

The following simple result, essentially due to \cite{D} (compare also \cite{PH}, Section 11.4), gives a general representation of the utility indifference price $p\left(H\right)$.

\begin{Prop}	\label{P4-3}
Let $U$ be a utility function and $H$ a positive $\F$-measurable contingent claim.  The utility indifference price $p\left(H\right)$ can be represented as follows:
\be	\label{P4-3-1}
p\left(H\right)=\frac{\E\left[U'\left(\bar{V}_T^{\,v,\pi^U}\right)\,\bar{H}\right]}
{\E\left[U'\left(\bar{V}_T^{\,v,\pi^U}\right)\,\bar{V}_T^{\,\pi^U}\right]}
\ee
\end{Prop}
\begin{proof}
Using equation \eqref{E4-4}, let us write the following Taylor's expansion:
\be	\label{P4-3-2}	\ba
W_p^U\left(\varepsilon\right)&=
\E\left[U\left(\bar{V}_T^{\,v,\pi^U}\right)
+\varepsilon\,U'\left(\bar{V}_T^{\,v,\pi^U}\right)\left(\bar{H}-p\,\bar{V}_T^{\,\pi^U}\right)
+o\left(\varepsilon\right)\right]	\\
&= W_p^U\left(0\right)+\varepsilon\,\E\left[U'\left(\bar{V}_T^{\,v,\pi^U}\right)\left(\bar{H}-p\,\bar{V}_T^{\,\pi^U}\right)\right]+o\left(\varepsilon\right)
\ea	\ee
If we insert \eqref{P4-3-2} into \eqref{E4-5} we get:
$$
\E\left[U'\left(\bar{V}_T^{\,v,\pi^U}\right)\left(\bar{H}-p\left(H\right)\,\bar{V}_T^{\,\pi^U}\right)\right]=0
$$
from which \eqref{P4-3-1} immediately follows.
\end{proof}

By combining Theorem \ref{T4-2} with Proposition \ref{P4-3}, we can easily prove the following Corollary, which yields an explicit and ``universal'' representation of the utility indifference price $p\left(H\right)$ (compare also \cite{RG}, Theorem 8, \cite{PH}, Section 11.4, and \cite{PR2}, Proposition 11). 

\begin{Cor}	\label{C4-1}
Let $H$ be a positive $\F$-measurable contingent claim. Then, under the assumptions of Theorem \ref{T4-2}, for any utility function $U$ the utility indifference price coincides with the real-world price (at $t=0$), namely:
$$
p\left(H\right)=\E\left[\frac{H}{V_T^{\,\pi^*}}\right]
$$
\end{Cor}
\begin{proof}
The present assumptions imply that, due to \eqref{T4-2-1}, we can rewrite \eqref{P4-3-1} as follows:
\be	\label{C4-1-1}
p\left(H\right)=
\frac{\E\left[U'\biggl(I\left(\frac{\mathcal{Y}\left(v\right)}{\bar{V}_T^{\,v,\pi^*}}\right)\biggr)\,\bar{H}\right]}
{\E\left[U'\biggl(I\left(\frac{\mathcal{Y}\left(v\right)}{\bar{V}_T^{\,v,\pi^*}}\right)\biggr)\,\bar{V}_T^{\,\pi^U}\right]}
= \frac{\E\left[\frac{\mathcal{Y}\left(v\right)}{\bar{V}_T^{\,v,\pi^*}}\,\bar{H}\right]}
{\E\left[\frac{\mathcal{Y}\left(v\right)}{\bar{V}_T^{\,v,\pi^*}}\,\bar{V}_T^{\,\pi^U}\right]}	
= \frac{\frac{1}{v}\,\E\left[\frac{\bar{H}}{\bar{V}_T^{\,\pi^*}}\right]}
{\frac{1}{v}\frac{\bar{V}_0^{\,\pi^U}}{\bar{V}_0^{\,\pi^*}}}	
= \E\left[\frac{H}{V_T^{\,\pi^*}}\right]
\ee
where the third equality uses the fact that $\pi^U\in\mathcal{A}^F$.
\end{proof}

\begin{Rem}
As can be seen from Definition \ref{D4-4}, the utility indifference price $p\left(H\right)$ depends a priori both on the initial endowment $v$ and on the chosen utility function $U$. The remarkable result of Corollary \ref{C4-1} consists in the fact that, under the hypotheses of Theorem \ref{T4-2}, the utility indifference price $p\left(H\right)$ represents an ``universal'' pricing rule, since it depends neither on $v$ nor on the utility function $U$ and, furthermore, it coincides with the real-world pricing formula.
\end{Rem}

\section{Conclusions, extensions and further developments}	\label{Concl}

In this work, we have studied a general class of diffusion-based models for financial markets, weakening the traditional assumption that the (NFLVR) no-arbitrage-type condition holds or, equivalently, that there exists an ELMM. We have shown that the financial market may still be viable, in the sense that arbitrages of the first kind are not permitted, as soon as the market price of risk process satisfies a crucial square-integrability condition. In particular, we have shown that the failure of the existence of an ELMM does not preclude the completeness of the financial market and the solvability of portfolio optimisation problems. Furthermore, in the context of a complete market, contingent claims can be consistently evaluated by relying on the real-world pricing formula.

We have chosen to work in the context of a multi-dimensional diffusion-based modelling structure since this allows us to consider many popular and widely employed financial models and, at the same time, avoid some of the technicalities which arise in more general settings. However, most of the results of the present paper carry over to a more general and abstract setting based on continuous semimartingales, as shown in Chapter 4 of \cite{Fo}. In particular, the latter work also deals with the robustness of the absence of arbitrages of the first kind with respect to several changes in the underlying modelling structure, namely changes of numéraire, absolutely continuous changes of the reference probability measure and restrictions and enlargements of the reference filtration.

The results of Section \ref{S4.3} on the valuation of contingent claims have been obtained under the assumption of a complete financial market. These results, namely that the real-world pricing formula \eqref{E4-1} coincides with the utility indifference price, can be extended to the more general context of an incomplete financial market, provided that we choose a logarithmic utility function.

\begin{Prop}	\label{T4-3}
Suppose that Assumption \ref{A1-2} holds. Let $H$ be a positive $\F$-measurable contingent claim such that $\E\left[\frac{\widehat{Z}_T}{S_T^0}\,H\right]<\infty$ and let $U\left(x\right)=\log\left(x\right)$. Then, the log-utility indifference price $p_{\,\log}\left(H\right)$ is explicitly given as follows:
$$
p_{\,\log}\left(H\right)=\E\left[\frac{H}{V_T^{\,\pi^*}}\right]
$$
\end{Prop}
\begin{proof}
Note first that $U\left(x\right)=\log\left(x\right)$ is a well-defined utility function in the sense of Definition \ref{D4-3}. Let us first consider Problem \eqref{ProbEUM} for $U\left(x\right)=\log\left(x\right)$. Using the notations introduced in the proof of Theorem \ref{T4-2}, the function $I$ is now given by $I\left(x\right)=x^{-1}$, for $x\in\left(0,\infty\right)$. 
Due to Proposition \ref{P2-1}, we have $\mathcal{W}\left(y\right)=v/y$ for all $y\in\left(0,\infty\right)$ and, hence, $\mathcal{Y}\left(v\right)=1$. Then, equation \eqref{T4-2-1} implies that $\bar{V}^{\,v,\pi^U}_T=\bar{V}^{\,v,\pi^*}_T$, meaning that the growth-optimal strategy $\pi^*\in\mathcal{A}^F$ solves Problem \eqref{ProbEUM} for a logarithmic utility function.
The same computations as in \eqref{C4-1-1} imply then the following:
$$
p_{\,\log}\left(H\right)
=\frac{\E\left[\frac{\bar{H}}{\bar{V}_T^{\,v,\pi^*}}\right]}
{\E\left[\frac{1}{\bar{V}_T^{\,v,\pi^*}}\bar{V}_T^{\,\pi^*}\right]}
=\E\left[\frac{H}{V_T^{\,\pi^*}}\right]
$$
\end{proof}

The interesting feature of Proposition \ref{T4-3} is that the claim $H$ does not need to be replicable. However, Proposition \ref{T4-3} depends on the choice of the logarithmic utility function and does not hold for a generic utility function $U$, unlike the ``universal'' result shown in Corollary \ref{C4-1}. Of course, the result of Proposition \ref{T4-3} is not surprising, due to the well-known fact that the growth-optimal portfolio solves the log-utility maximisation problem, see e.g. \cite{Be}, \cite{CL} and \cite{KK}.

\begin{Rem}
Following Section 11.3 of \cite{PH}, let us suppose that the discounted GOP process $\bar{V}^{\,\pi^*}=\left(\bar{V}^{\,\pi^*}_t\right)_{0\leq t \leq T}$ has the Markov property under $P$. Under this assumption, one can obtain an analogous version of Theorem \ref{T4-2} also in the case of an incomplete financial market model (see \cite{PH}, Theorem 11.3.3). In fact, the first part of the proof of Theorem \ref{T4-2} remains unchanged. One then proceeds by considering the martingale $M=\left(M_t\right)_{0\leq t \leq T}$ defined by $M_t:=\E\bigl[\widehat{Z}_T\,I\bigl(\mathcal{Y}\left(v\right)/\bar{V}_T^{\,v,\pi^*}\bigr)\bigl|\F_t\bigr]=\E\bigl[1/\bar{V}_T^{\,\pi^*}\,I\bigl(\mathcal{Y}\left(v\right)/\bar{V}_T^{\,v,\pi^*}\bigr)\bigl|\F_t\bigr]$, for $t\in\left[0,T\right]$.
Due to the Markov property, $M_t$ can be represented as $g\bigl(t,\bar{V_t}^{\,\pi^*}\bigr)$, for every $t\in\left[0,T\right]$. 
If the function $g$ is sufficiently smooth one can apply It\^o's formula and express $M$ as the value process of a benchmarked fair portfolio. If one can shown that the resulting strategy satisfies the admissibility conditions (see Definition \ref{D1-1}), Proposition \ref{P4-3} and Corollary \ref{C4-1} can then be applied to show that the real-world pricing formula coincides with the utility indifference price (for any utility function!). Always in a diffusion-based Markovian context, a detailed analysis to this effect can also be found in the recent paper \cite{Ru}.
\end{Rem}

We want to point out that the modeling framework considered in this work is not restricted to stock markets, but can also be applied to the valuation of fixed income products. In particular, in \cite{BLNSP} and \cite{PH}, Section 10.4, the authors develop a version of the Heath-Jarrow-Morton approach to the modeling of the term structure of interest rates without relying on the existence of a martingale measure. In this context, they derive a real-world version of the classical Heath-Jarrow-Morton drift condition, relating the drift and diffusion terms in the system of SDEs describing the evolution of forward interest rates. Unlike in the traditional setting, this real-world drift condition explicitly involves the market price of risk process.  

Finally, we want to mention that the concept of real-world pricing has also been studied in the context of incomplete information models, meaning that investors are supposed to have access only to the information contained in a sub-filtration of the original full-information filtration $\FF$, see \cite{RG}, \cite{PR1} and \cite{PR2}.

\vspace{1.5cm}
\noindent
\textbf{Acknowledgements:}
Part of this work has been inspired by a series of research seminars organised by the second author at the Department of Mathematics of the Ludwig-Maximilians-Universit\"at M\"unchen during the Fall Semester 2009. The first author gratefully acknowledges financial support from the ``Nicola Bruti-Liberati'' scholarship for studies in Quantitative Finance. We thank an anonymous referee for the careful reading and for several comments that contributed to improve the paper.

\end{document}